\documentclass[sigconf,screen,nonacm]{acmart}
\AtBeginDocument{%
  \providecommand\BibTeX{{%
    \normalfont B\kern-0.5em{\scshape i\kern-0.25em b}\kern-0.8em\TeX}}}

\usepackage[utf8]{inputenc} 
\usepackage{stmaryrd}  
\usepackage{xspace}    
\usepackage{xcolor}    

\newcommand{\new}[1][]{\reflectbox{\sf{#1}N}}
\newcommand{\Var}{\mathcal{V}ar}
\newcommand{\Sort}{S}
\newcommand{\NSort}{NS}
\newcommand{\DSort}{DS}
\newcommand{\ASort}{AS}
\newcommand{\sort}{\tau}
\newcommand{\nsort}{\alpha}

\newcommand{\defeq}{\stackrel{\triangle}{=}}
\newcommand{\swap}[3]{(#1~#2)\cdot #3}
\newcommand{\abs}[2]{[#1]#2}
\newcommand{\fresh}{\mathrel{\#}}
\newcommand{\aprolog}{$\alpha$Prolog\xspace}
\newcommand{\supp}{\mathit{supp}}
\newcommand{\bbA}{\mathbb{A}}
\newcommand{\bbK}{\mathbb{K}}
\newcommand{\Aa}{\mathsf{a}}
\newcommand{\Ab}{\mathsf{b}}
\newcommand{\Ac}{\mathsf{c}}
\newcommand{\Ad}{\mathsf{d}}

\newcommand{\lam}{\textsf{lam}}
\newcommand{\app}{\textsf{app}}
\newcommand{\subst}{\textsf{subst}}

\newcommand{\df}[1]{\left \lceil #1 \right \rceil }

\newcommand{\coerce}[1]{#1^\dagger}
\renewcommand{\vec}[1]{\bar{#1}}
\newcommand{\sem}[1]{\|#1\|}
\newcommand{\semantics}[2]{\sem{#2}_{#1}}
\newenvironment{arcl}{\[\begin{array}{rcl}}{\end{array}\]}

\begin{document}

\title{Nominal Matching Logic}

\author{James Cheney}
\affiliation{%
  \institution{Laboratory for Foundations of Computer Science, University of Edinburgh }
  \city{Edinburgh}
  \country{United Kingdom}}
\email{jcheney@inf.ed.ac.uk}

\author{Maribel Fern\'andez}
\affiliation{%
  \institution{Department of Informatics,\\  King's College London }
  \city{London}
  \country{United Kingdom}}
\email{Maribel.Fernandez@kcl.ac.uk}

\renewcommand{\shortauthors}{Cheney and Fern\'andez}

\begin{abstract}
We introduce Nominal Matching Logic (NML) as an extension of Matching Logic with names and binding following the Gabbay-Pitts nominal approach. Matching logic is the foundation of the $\bbK$ framework,  used to specify  programming languages and automatically derive associated tools (compilers,  debuggers,  model checkers, program verifiers). Matching logic does not include a primitive notion of name binding, though binding operators can be represented via an encoding that internalises the graph of a function from bound names to expressions containing bound names. This approach is sufficient to represent computations involving binding operators, but  has not been reconciled with support for inductive reasoning over syntax with binding (e.g., reasoning over $\lambda$-terms). Nominal logic  is  a formal system for reasoning about names and binding, which provides well-behaved and powerful principles for inductive reasoning over syntax with binding, and NML inherits these principles.  We discuss design alternatives for the syntax and the semantics of NML, prove meta-theoretical properties and give examples to illustrate its expressive power. In particular, we show how induction principles for $\lambda$-terms ($\alpha$-structural induction) can be defined and used to prove standard properties of the $\lambda$-calculus.
\end{abstract}



\maketitle

\section{Introduction}\label{sec:intro}

Nominal logic is a formal system for reasoning about names and name-binding, building on the Gabbay-Pitts approach to name-binding founded on the concepts of permutation groups acting on sets, freshness, and finite support~\cite{pitts:nomlfo-jv,gabbay:newaas-jv}.  Its semantics and proof theory have been extensively investigated and provide a foundation for nominal extensions to rewriting~\cite{FernandezM:nomr-jv}, functional programming~\cite{shinwell03icfp}, logic programming~\cite{cheney08toplas}, and specification testing~\cite{cheney17tplp}.  

For example, the syntax of the $\lambda$-calculus can be specified using symbols $var$, $lam$ and $app$ (to represent variables, $\lambda$-abstraction and application respectively); a term of the form $\lambda x.e$ is represented as $lam([x]\textbf{e})$ where $\textbf{e}$ is the representation of $e$ and $[x]\textbf{e}$ is a nominal abstraction. The $\alpha$-equivalence relation between $\lambda$-terms holds by construction, thanks to the properties of the abstraction construct in nominal logic. Capture-avoiding substitution  can be formally defined in nominal logic in a  way that resembles the familiar informal definition (see Example~\ref{ex:lamc}).

Nominal logic also provides well-behaved and powerful principles for inductive reasoning over abstract syntax with binding, as supported in Nominal Isabelle~\cite{UrbanC:ntih}. 
However, using nominal logic for specification and reasoning is sometimes awkward due to the partial or nondeterministic nature of some operations involving names: for example the nondeterministic \emph{fresh name} operation which chooses a fresh name, or the partial \emph{concretion} operation which renames a bound, abstracted name to a fresh one.  These operations are difficult to work with directly in first-order logic where all function symbols must denote total, deterministic functions.

Matching logic is a formal system for reasoning about patterns, which may fail to match (partiality) or match multiple ways (nondeterminism).  It has been introduced as a formal foundation for modeling and reasoning about operational semantics in the language framework $\bbK$~\cite{RosuG:K}, which has been used to develop semantics for a number of programming languages.

Matching logic does not contain built-in support for binding, which limits its appeal for modeling high-level languages, such as functional programming languages based on the $\lambda$-calculus.  
Recently, matching logic has been extended with facilities for modelling and reasoning about binding by \citet{chen20icfp}, for example to capture the equational theory of the lambda-calculus or model other binding constructs.  In essence, the approach taken is to define binding by internalizing the notion of the graph of a function (e.g. a mapping from bound names to expressions that might contain the bound name).  However, as Chen and Rosu note, it is not yet clear how to extend this approach further to support inductive reasoning over syntax with binding.  

As noted above, both nominal logic and matching logic have strengths and weaknesses. Nominal logic provides support for inductive reasoning over syntax with binding, but  reasoning about instantiating abstractions with specific choices of names, or choices of fresh names, often involves partiality or nondeterminism issues that necessitate abandoning a calculational style of proof in nominal logic. On the other hand, matching logic provides support for partiality and non-determinism but does not support inductive reasoning over syntax with binding. Since nominal logic is definable as an ordinary first-order theory, it may be considerably easier to incorporate into matching logic and implementations such as $\bbK$ than alternative approaches to binding.

In this paper we show that matching logic and nominal logic are indeed compatible, and show how they can be combined in a single system that generalizes each of them.  Thus, a wealth of existing nominal techniques for reasoning about languages with names and binding can be imported into matching logic. 
We discuss two approaches to combine matching logic and nominal logic. In the first approach, the axioms of nominal logic are specified as a theory in matching logic, obtaining a system we call NLML. We show that NLML has the expected properties, in particular, it allows us to use  concepts such as freshness of names and abstraction in matching logic specifications. Morever,  existing proof systems and implementations of matching logic can be used to reason or execute specifications with binding operators. However, NLML does not include a primitive notion of name, and the $\new$ quantifier of nominal logic cannot be used in NLML patterns. To address these shortcomings, in the second approach we extend the syntax of matching logic with a distinguished category of names and a $\new$ constructor for patterns, and provide a semantics for these extensions that is compatible with the standard matching logic semantics. This also avoids the need to include nominal logic as a theory in matching logic. 

We first review basic concepts of nominal logic and matching logic and their semantics (Sec.~\ref{sec:background}).  We then briefly present NLML before introducing
\emph{Nominal Matching Logic} (NML), a single system that incorporates the capabilities of nominal logic for specifying and reasoning about name-binding into matching logic (Sec.~\ref{sec:nml}).  We illustrate the applications of NML via standard lambda-calculus examples (Sec.~\ref{sec:examples}).  
The simplicity of this approach offers potential advantages over 
previous approaches to specify binders in matching logic,
for example, there is no theoretical obstacle to adding fixed point operators to NML to support reasoning by induction about structures involving names and binding.  To substantiate this point we provide a detailed comparison with Chen and Rosu's approach to binding in Applicative Matching Logic (AML) (Sec.~\ref{sec:comparison}).  

The question of how to model and reason about binding in formalizing programming languages, logics and calculi has a long history, and nominal abstract syntax is but one of many approaches.  This paper focuses somewhat narrowly on the question how to augment matching logic with support for binding, and our claim is that nominal logic (with some adjustment) is a promising candidate for doing this.  We do not claim that this sets a new standard for reasoning about binding syntax that is superior to existing systems such as Nominal Isabelle~\cite{UrbanC:ntih}, locally nameless~\cite{PollackR:metatheory}, Abella~\cite{abella}, Beluga~\cite{beluga}, etc., and it seems premature to compare our work so far (which is mostly theoretical and not yet implemented) with established, mature systems that are already widely used for formalizing metatheory.  We only claim that if we wish to extend matching logic with binding, nominal logic seems better suited than other approaches.  We discuss how our work fits into the broader landscape in more detail in Section~\ref{sec:related}.

All proofs are relegated to the appendix.

\section{Background}\label{sec:background}
We take for granted familiarity with standard sorted first-order logic syntax and semantics. In this section we recall the syntax and semantics of matching logic and nominal logic. In both cases, a signature $\mathbf{\Sigma} = (\Sort,\Var,\Sigma)$ specifies $\Sort$ a set of sorts, $\Var$ a sort-indexed family of countably many variables for each sort, and $\Sigma$ a family of sets $\Sigma_{\sort_1,\ldots,\sort_n;\sort}$ indexed by $\Sort^* \times (\Sort \cup \{Pred\})$, assigning symbols $\sigma$ their input and output sorts.  For convenience and consistency between nominal and matching logic we assume that sorts are closed under finite products, that is, $\Pi_{i=1}^n \sort_i $ is a sort whenever $\sort_1,\ldots,\sort_n$ are sorts, with associated function symbols for constructing $n$-tuples $(\_,\ldots,\_) \in \Sigma_{\sort_1,\ldots,\sort_n;\Pi_{i=1}^n \sort_i}$ and for projecting $\pi_j^n : \Pi_{j=1}^{n} \sort_j; \sort_i$.  In first-order (and nominal) logic $\Sigma$ also needs to specify the sorts of relation symbols $p$, and we use the special sort $Pred$ and write $p \in \Sigma_{\sort_1,\ldots,\sort_n;Pred}$ to indicate this.

\subsection{Nominal Logic}\label{sec:nl}
Nominal logic was introduced by Pitts as a sorted first-order theory~\cite{pitts:nomlfo-jv}.  In nominal logic, signatures are augmented by distinguishing some sorts as \emph{name sorts}, and sorts include a construction called \emph{abstraction} that builds a new sort $[\nsort]\sort$ from a name sort $\nsort$ and any sort $\sort$.  Nominal logic in general allows for the possibility of multiple name sorts, but we will consider the case of a single name-sort for simplicity, since handling the general case requires additional bureaucracy that is not needed to illustrate the main points.  Abstractions in $\abs{\nsort}{\sort}$ correspond to elements of $\sort$ with a distinguished \emph{bound} name.  The signature of any instance of nominal logic includes function symbols  $\swap{-}{-}{-} : \nsort \times \nsort \times \sort \to \sort$ and $\abs{-}{-} : \nsort \times \sort \to \abs{\nsort}{\sort}$ denoting name swapping and abstraction, for any name sort $\nsort$ and sort $\sort$, and atomic formulas for equality $(=)$ at any sort and freshness $(\fresh)$ relating any name sort and any sort.  Finally, nominal logic includes an extra quantifier, $\new a. \phi$, which is pronounced "for fresh $a$, $\phi$ holds."

Pitts gives a set of axioms and axiom schemas that describe the behavior of these constructs, see Figure~\ref{fig:axiNL}.  The $S$ axioms describe the behavior of swapping, while axioms $F$ and $A$ characterize freshness and abstraction respectively.  The $E$ axioms ensure equivariance, that is, that swapping preserves the behavior of function and atomic predicate symbols.  Finally axiom scheme $(Q)$ characterizes the $\new$-quantifier.  Aside from $\new$, the axioms are standard first-order axioms, and $(Q)$ can be viewed as a recipe for defining $\new$ in terms of other formulas.  Rather than discuss each axiom in detail, we review the semantics of nominal logic which validates them.

\begin{figure}
\[\begin{array}{cl}
\swap{a}{a}{x} = x & (S1)\\
\swap{a}{a'}{\swap{a}{a'}{x}} = x & (S2)\\
\swap{a}{a'}{a} = a' & (S3)\\
\swap{a}{a'}{\swap{b}{b'}{x}} = \swap{(\swap{a}{a'}{b})}{(\swap{a}{a'}{b'})}{\swap{a}{a'}{x}} & (E1)\\
b \fresh x \Rightarrow \swap{a}{a'}{b} \fresh \swap{a}{a'}{x} & (E2)\\
\swap{a}{a'}{f(\vec{x})} = f(\swap{a}{a'}{\vec{x}}) & (E3)\\
p(\vec{x})  \Rightarrow p(\swap{a}{a'}{\vec{x}}) & (E4)\\
\swap{b}{b'}{\abs{a}{x}} = \abs{\swap{b}{b'}{a}}{\swap{b}{b'}{x}} & (E5)\\ 
a \fresh x \wedge a' \fresh x \Rightarrow \swap{a}{a'}{x} = x & (F1)\\
a \fresh a' \iff a \neq a' & (F2)\\
\forall a:\nsort,a':\nsort'. a \fresh a' \qquad (\text{$\nsort \neq \nsort'$})& (F3) \\
\forall \vec{x}. \exists a. a \fresh \vec{x} & (F4)\\
\forall \vec{x}.(\new a. \phi \iff \exists a. a \fresh \vec{x} \wedge \phi)     \qquad (FV(\new a.\phi) \subseteq \vec{x})& (Q)\\
\abs{a}{x} = \abs{a'}{x'} \iff (a = a' \wedge x = x') \\
\qquad \qquad \qquad \qquad \qquad \qquad \vee (a \fresh x' \wedge \swap{a}{a'}{x} = x') & (A1)\\
\forall x:\abs{\nsort}{\sort}.\exists a:\nsort,y:\sort. x = \abs{a}{y} & (A2)
\end{array}\]
\caption{Axioms of (classical) Nominal Logic}
\label{fig:axiNL}
\end{figure}

Nominal sets (cf.~\cite{pitts:ns}) provide semantics to nominal logic. 
The characteristic feature of nominal sets is the use of \emph{name permutations} acting on sets.  Let $G$ be the symmetric group $Sym(\mathbb{A})$ on some countable set $\mathbb{A}$ of atoms, which is generated by the swappings $(a~b)$ where $a,b \in \bbA$. 
A \emph{$G$-set} is a structure $(X,\cdot)$ equipped with carrier set $X$ and an action of $G$ on $X$, i.e. a function $\pi,x \mapsto \pi \cdot x$ satisfying the laws (1) $id \cdot x = x$ and (2) $(\pi \circ \pi') \cdot x = \pi \cdot \pi' \cdot x$.
The \emph{support} of an element of $G$ is the set of atoms not fixed by $\pi$, i.e. $\supp(\pi) = \{a\in \mathbb{A} \mid \pi(a) \neq a\}$.  
A function or relation on $G$-sets is called \emph{equivariant} if it commutes with the permutation action: $\pi \cdot f(\vec{x}) = f(\pi \cdot \vec{x})$ or $R(\vec{x}) = R(\pi\cdot \vec{x})$.  

$G$-sets and equivariant functions form a topos and so admit many standard constructions familiar from set theory, including products, sums (disjoint union), exponentials (function spaces), power sets, etc.  For example, products of  $G$-sets are formed by taking products of the underlying carrier sets and extending the group action pointwise.  

A \emph{support} of an element $x$ of a $G$-set is a set $S$ of atoms such that for all $\pi$ with $\supp(\pi) \cap S = \emptyset$
we have $\pi\cdot x = x$. 
A \emph{nominal set} is a $G$-set in which all elements have a finite support.  It is routine to show that this implies each element has a unique, least finite support which is denoted $\supp(x) = \bigcap \{S \subseteq \mathbb{A} \mid \forall \pi.\supp(\pi) \cap S = \emptyset \Rightarrow \pi\cdot x = x\}$.  Nominal sets also form a topos~\cite{pitts:ns}, so also admit most standard set-theoretic constructions, except that functions and power sets must also be finitely supported (i.e. $X\to_{fin} Y$ and $\mathcal{P}_{fin}(X)$ consist of all functions/subsets having finite support).  

Two special constructions on nominal sets are the set $\mathbb{A}$ of atoms, where $\pi \cdot a = \pi(a)$ and $\supp(a) = \{a\}$, and the set of abstractions over $\mathbb{A}$ and $X$, written $\abs{\mathbb{A}}{X}$.  The latter is defined as the set of equivalence classes of pairs $\mathbb{A} \times X$ by the following relation:
\[\langle a,x\rangle \equiv_\alpha \langle b, y\rangle \iff \forall c \notin\supp(a,b,x,y). \swap{a}{c}{x} = \swap{b}{c}{y}\]
Swapping applies to such equivalence classes pointwise, that is, $\pi\cdot E = \{\langle \pi\cdot a, \pi \cdot x \rangle \mid \langle a,x \rangle \in E\}$, and $\supp(\abs{a}{x}) = \supp(x) - \{a\}.$
It is again a standard result that $\mathbb{A}$ and $\abs{\mathbb{A}}{X}$, as defined above, are nominal sets if $X$ is.  

Abstract syntax with binding operators can be directly represented using the \emph{abstraction} construct of nominal logic as shown in the example below for the $\lambda$-calculus.

\begin{example}
\label{ex:lamc}
To represent the syntax of the $\lambda$-calculus, it is sufficient to use a signature including function symbols $var$, $lam$ and $app$ (to represent variables, $\lambda$-abstraction and application respectively), and  sorts $Var, Exp$ where $Var$ is a name sort and $Exp$ is  equipped with the following constructors:
\[\small \begin{array}{c}
var\colon Var \to Exp \quad app\colon Exp \times Exp \to Exp \quad lam\colon [Var]Exp \to Exp\end{array}\]
A term of the form $\lambda x.e$ is represented as 
$lam([x]\textbf{e})$ where $\textbf{e}$ is the representation of $e$. The $\alpha$-equivalence relation between $\lambda$-terms holds by construction, thanks to the properties of the abstraction construct in nominal logic.  Likewise, for example, the substitution operation  $subst\colon Exp\times Var \times Exp \rightarrow Exp$ can be axiomatized as follows:
\[\begin{array}{cl}
subst(var(x),x,z)  = z & (Subst1)\\
x \neq y \Rightarrow subst(var(x),y,z) = var(x) & (Subst2)\\
subst(app(x_1,x_2),y,z) = \\
\qquad app(subst(x_1,y,z),subst(x_2,y,z)) & (Subst3)\\
a \fresh y,z \Rightarrow subst(lam(\abs{a}{x}),y,z) = \\
\qquad lam(\abs{a}{subst(x,y,z)}) & (Subst4)
\end{array}\] 
Here, the first three equations are ordinary first-order (conditional) equations, while the fourth specifies how substitution behaves when a $\lambda$-bound name is encountered: the name in the abstraction is required to be fresh, and when that is the case, the substitution simply passes inside the bound name.  Even though the $(Subst4)$ axiom is a conditional equation, the specified substitution operation is still a total function, because any abstraction denotes an alpha-equivalence class that contains at least one sufficiently fresh name so that the freshness precondition holds.  Other cases where the name is not sufficiently fresh, or where we attempt to substitute for the bound name, do not need to be specified, instead they can be proved from the above axioms and an induction principle for $\lambda$-terms (see e.g. \citet{UrbanC:ntih} or Section~\ref{sec:examples}).
\end{example}

In the original form of nominal logic introduced by Pitts, there are no constants of ``name'' sorts: in fact, adding such constants makes nominal logic inconsistent, thanks to the equivariance axiom~$(E3)$ (considering constants as a special case of 0-ary functions).
The first system we will present, NLML (Section~\ref{sec:nomlm}) is based on directly importing nominal logic into matching logic and so in NLML it is also the case that there are no name constants.
However, having syntax for ground names (sometimes called ``atoms'' or ``name constants'') is attractive in many practical situations, such as nominal unification, rewriting, and  logic programming.  For example, in standard nominal unification and matching algorithms, certain arguments such as $a,b$ in $\swap{a}{b}{t}$, $a \fresh t$ and $\abs{a}{t}$, are required to be ground, and this is critical for efficiency and for most general unifiers to exist.  Also, in logic programming the semantics of programs is defined in terms of \emph{Herbrand models} built out of ground terms, and this does not work for Pitts' formulation of nominal logic because there are no ground terms involving names.  Nominal logic was modified by \citet{cheney:comhtn,cheney16jlc} to allow for ground names $\Aa,\Ab,\ldots$ as an additional syntactic class which can be bound by the $\new$-quantifier.  We will adopt a similar approach later in the paper in the NML system (Section~\ref{sec:newnml}).

\subsection{Matching Logic}
There are several different recent presentations of variations of matching logic~\cite{chen19lics,chen20icfp,rosu-2017-lmcs}.  Our presentation is closest to that of \citet{chen19lics}, though we believe our approach could be adapted to other systems such as Applicative Matching Logic (AML)~\cite{chen20icfp} without problems.  The most important difference between these systems is that in AML, there are no sorts governing the syntax, and instead sorts are internalized into the logic as constants.  The presentation of \citet{chen19lics} uses a standard multisorted syntax and this presentation is closer to Pitts's presentation of Nominal Logic, so we follow this approach.

We consider a matching logic signature  $\mathbf{\Sigma} = (\Sort,\Var,\Sigma)$ as above, which we assume contains operations for product sorts and $Pred$.

In matching logic, unlike (say) first-order logic, there is only one syntactic class of \emph{patterns}:
\[
\phi_\sort ::= x\colon \sort \mid \phi_\sort \wedge \psi_\sort \mid \neg \phi_\sort \mid \exists x\colon\sort'. \phi_\sort \mid \sigma(\phi_{\sort_1},\ldots,\phi_{\sort_n})\]
where in the first case $x\in \Var_\sort$ and in the last case $\sigma \in \Sigma_{\sort_1,\ldots,\sort_n;\sort}$.  
Patterns generalize both terms and first-order logic formulas.  We describe their semantics informally before providing the formal semantics.  The meaning of a pattern of sort $\sort$ is a set of elements of sort $\sort$ that the pattern matches.  The variable pattern $x\colon\sort$ matches just one element, the value of $x$.  Conjunction corresponds to intersection of matching sets and negation corresponds to complementation.  Existential patterns $\exists x\colon\sort'. \phi_\sort$ match those elements of $\sort$ for which there exists a value of $x$ of sort $\sort'$ making the pattern $\phi_\sort$ match.  Finally, symbol patterns $\sigma(\phi_{\sort_1},\ldots,\phi_{\sort_n})$ match as follows.  Each symbol is interpreted as an arbitrary relation between $n$ input arguments and an output (or equivalently, as a function from the inputs to a set of possible matches).  Note that such symbols do not necessarily correspond to functions: they may be undefined, and there may be multiple matching values for a given choice of inputs. 
We assume there is a distinguished sort $Pred$ and unary symbols $\coerce{(-)}_\sort \in \Sigma_{Pred;\sort}$ for every sort $\sort$, called \emph{coercion}, as well as equality symbols $=_\sort$  at every sort $\sort$.  
Other patterns (disjunction $\phi_1 \vee \phi_2$, implication $\phi_1 \Rightarrow \phi_2$, universal quantification, true and false) can be defined as abbreviations, for example, $\top_\sort \defeq \exists x\colon \sort.x\colon \sort$ (since carriers are non-empty sets) and $\bot_\sort \defeq \neg \top_\sort$.  Subscripts indicating the sort at which an operation is used are omitted when clear from context.

Formally, the semantics of patterns is defined as follows.
A matching logic model $M =(\{M_\sort\}_{\sort \in S}, \{\sigma_M\}_{\sigma \in \Sigma})$
consists of a non-empty carrier set $M_\sort$ for each sort $\sort \in S$ and an interpretation $\sigma_M\colon M_{\sort_1}\times\ldots\times M_{\sort_n}\rightarrow \mathcal{P}(M_\sort)$ for each $\sigma \in \Sigma_{\sort_1,\ldots,\sort_n;\sort}$.  A \emph{valuation} $\rho$ is a function  $\rho \colon \Var\to  M$ that  respects sorts so that $\rho(x) \in M_\sort$ for each $x \in \Var_\sort$.

Given $\mathbf{\Sigma} = (S,\Var,\Sigma)$, a matching logic $\Sigma$-model $M$ and valuation $\rho$, the extension $\semantics{\rho}{-}$ to patterns is defined by:
$\semantics{\rho}{x} = \{\rho(x)\}$ for all $x \in \Var$, 
$\semantics{\rho}{\phi_1 \wedge \phi_2} = \semantics{\rho}{\phi_1} \cap  \semantics{\rho}{\phi_2}$,
$\semantics{\rho}{\neg\phi_\sort} = M_\sort - \semantics{\rho}{\phi_\sort}$,
$\semantics{\rho}{\exists x\colon \sort'. \phi_\sort} = \bigcup_{a \in M_{\sort'}} \semantics{\rho[a/x]}{\phi_\sort}$,
$\semantics{\rho}{\sigma(\phi_{\sort_1},\ldots, \phi_{\sort_n})} = \overline{\sigma_M}(\semantics{\rho}{\phi_{\sort_1}},\ldots,\semantics{\rho}{\phi_{\sort_n}})$, for  $\sigma \in \Sigma_{\sort_1,\ldots,\sort_n;\sort}$, where we write $\overline{\sigma_M}$ to denote the pointwise extension of $\sigma_M$, i.e. $\overline{\sigma_M}(V_1,\ldots,V_n) = \bigcup\{\sigma_M(v_1,\ldots,v_n) \mid v_1\in V_1,\ldots,v_n\in V_n\}$.

A pattern $\phi_\sort$ is valid in $M$, written $M \vDash \phi$, if $\semantics{\rho}{\phi_\sort} = M_\sort$ for all valuations $\rho\colon \Var \rightarrow M$. 
If $\Gamma$ is a set of patterns (called axioms), then $M \vDash \Gamma$ if $M \vDash \phi$ for each $\phi \in \Gamma$ and $\Gamma \vDash \phi$ if $M \vDash \phi$ for all $M \vDash \Gamma$. The pair $(\mathbf{\Sigma},\Gamma)$ is a matching logic theory, and $M$ is a model of the theory if $M \vDash \Gamma$.

Again following \citet{chen19lics} we assume the sort $Pred$ is interpreted as a single-element set $M_{Pred} = \{\star\}$, thus, a predicate is true if its interpretation is $M_{Pred}$ and false if its interpretation is $\emptyset$.  The interpretation of the  operation $\coerce{(-)} \in \Sigma_{Pred;\sort}$  maps $\star$ to $M_\sort$.  That is, a true predicate is interpreted as $\{\star\}$ which can be coerced to the set of all patterns matching some other sort, while a  false predicate is interpreted as $\emptyset$ which can be coerced to the empty pattern of sort $\sort$.  Coercions may be omitted when obvious from context.
Equality $(=_\sort)$ predicates are  interpreted as follows: $\semantics{\rho}{\phi = \psi} =  \{\star \mid \semantics{\rho}{\phi} = \semantics{\rho}{\psi}\}$.  We may write $\phi \subseteq \psi$ as an abbreviation for $\phi \vee \psi = \psi$, and $x \in \phi$ for $x \subseteq \phi$ to emphasize that when $x$ is a variable it matches exactly one value.  Some presentations of matching logic include other primitive formulas that are definable using equality and coercion, such as \emph{definedness} $\df{\phi}_\sort$ (which can be defined as $\exists x:\sort. x \coerce{\in} \phi$). 
Alternatively, the coercion operator can be defined using a primitive definedness symbol.
The notation $\sigma\colon \sort \to \sort'$ (resp.\ $\sigma\colon \sort\rightharpoonup \sort'$) indicates that  $\sigma$ is a \emph{function} (resp.\ \emph{partial function}), in which case an appropriate axiom for $\sigma$ is assumed to be included in the theory.

Variables in patterns can be substituted by patterns: $\phi[\psi/x]$ denotes the result of substituting $\psi$ for every free occurrence of $x$ in $\phi$, where $\alpha$-renaming happens implicitly to prevent variable capture. Substituting a pattern for a universally-quantified variable does not preserve validity in general, however, if $\phi$ is valid and $\psi$ is functional, i.e., it evaluates to a singleton set, then  $(\forall x. \phi) \Rightarrow \phi[\psi/x]$  is valid  (see~\cite{chen19lics}). 

Hilbert-style proof systems for matching logic are available: we refer to \cite{chen19lics,rosu-2017-lmcs} for details.

\section{Nominal Matching Logic}\label{sec:nml}

We now consider  ways to combine nominal logic and matching logic, yielding a logical framework extending the advantages of matching logic for reasoning about language semantics with the advantages of nominal logic for reasoning about name binding.

\subsection{Nominal Logic as a Matching Logic Theory}
\label{sec:nomlm}
The first approach is simply to consider an instance of nominal logic (i.e., a signature and theory) as an instance of matching logic.    Since any (sorted) first-order logic theory can be translated to a theory of matching logic, and any instance of nominal logic is a (sorted) first-order logic theory, it is clear we can define nominal logic as a theory in matching logic. However, some adjustments are needed, as with any translation of first-order logic into matching logic: we must specify that the function symbols of the nominal logic signature satisfy the \emph{Function} axiom (i.e., denote functions)  and predicate symbols are pattern symbols with result sort $Pred$.  Also, all the axioms  of nominal logic must be included as translated axioms in the matching logic theory.
  This is not as straightforward as it might sound, since the nominal logic axioms treat function symbols and relation symbols differently, whereas in matching logic there is no built-in distinction.  First, we define the syntax  of the nominal logic theory $NL$ in matching logic.

\subsubsection{Syntax}
The nominal logic signature includes function symbols to represent swappings and abstraction, which map to symbols in the matching logic signature, interpreted using the axioms in nominal logic. Similarly, the equality and freshness relation symbols from nominal logic map to  symbols in matching logic with appropriate interpretations.

Concretely, let $\Sigma^{NL}$ be a nominal logic signature, thus including $\swap{-}{-}{-} : \nsort \times \nsort \times \sort \to \sort$ and $\abs{-}{-} : \nsort \times \sort \to \abs{\nsort}{\sort}$ (name swapping and abstraction symbols, for any name sort $\nsort$ and sort $\sort$) and equality  $= :  \sort \times \sort$ at any sort $\sort$ and freshness $\fresh : \nsort \times \sort$ for any name sort $\nsort$ and sort $\sort$.   We define the matching logic signature $\Sigma^{NLML}$ by using the predicate sort $Pred$ and considering all relation symbols $R : \sort_1\times\cdots\times \sort_n$ to be matching logic symbols in $\Sigma_{\sort_1\times\cdots\times \sort_n; Pred}$.  Otherwise, constants and function symbols of nominal logic are regarded as matching logic symbols in the obvious way.  To ensure that they are interpreted as (total, deterministic) functions in matching logic we also add instances of the function axioms:
\[\begin{array}{clccl}
\exists x. c = x & (\mathit{Fun}_c) & ~~~~ &
\forall \vec{z}.\exists x. f(\vec{z}) = x & (\mathit{Fun}_f)
\end{array}\]

Any formula of the form $\new a. \phi$ in nominal logic can be replaced by $\exists a. a\fresh \vec{x} \wedge \phi$ preserving its truth value (axiom (Q) justifies this translation, see Figure~\ref{fig:axiNL}). Therefore we do not consider  formulas with $\new$ quantifiers yet (we revisit this decision in Section~\ref{sec:newnml}).  

\subsubsection{Axioms}  
All the nominal logic axioms (see Figure~\ref{fig:axiNL}) may be viewed as matching logic axioms; however since there is no distinction between constants, function symbols and relation symbols in matching logic we can replace the various equivariance axiom schemes with a single scheme stating that all symbols in $\Sigma^{NLML}$ are equivariant:
\[\swap{a}{b}{\sigma(\vec{x})} = \sigma(\swap{a}{b}{\vec{x}}) \qquad (EV)\]
There is a minor subtlety in dealing with equivariance for predicates, since for a predicate the above axiom scheme does not ensure that $R(\vec{x}) \Rightarrow R(\swap{a}{b}{\vec{x}})$.  To recover axioms $(E2),(E4)$ of nominal logic, we also include an axiom ensuring that there is just one (necessarily equivariant) value of sort $Pred$ (note that $x$ in this axiom can only be interpreted by the single element of $M_{Pred}$):
\[\forall x\colon Pred. x = \top_{Pred} \qquad (P)\]
Note that although we assume $Pred$ is interpreted as a single-element set in all models, we need this axiom in order to use this fact in proofs.

 Concretely the axioms are: 
\[\begin{array}{rcl}
Ax_{NLML} &= &\{(S1),(S2),(S3), (EV), (P), (F1),(F2),(F3),(F4), \\
&&(A1), (A2), (Fun_f), (Fun_c)\}\end{array} \]
where $f$ and $c$ are the function symbols and constants of the nominal logic signature.
We write $NLML$ to refer to matching logic instantiated by (translated) nominal logic axioms as above.

\subsubsection{Models}
Let $M$ be a matching logic  $\Sigma_{NLML}$-model satisfying the axioms in $Ax_{NLML}$.
Axiom $(F4)$ ensures that there are infinitely many atoms in the carrier sets of name sorts. 

True in nominal logic maps to  $\top$ of sort $Pred$, matching every element of the domain.
False in nominal logic maps to $\bot$ of sort $Pred$, matching nothing.

Axioms $(S1)$--$(S3)$, $(EV)$, $(P)$ ensure that the interpretation of the swapping symbol $\swap{-}{-}{-} \in \Sigma_{NLML}$ is a swapping operator: $(S1)$--$(S3)$ define swapping for atoms, and $(EV)$, $(P)$ define swapping for the interpretation of patterns built using function symbols or relation symbols.
Axioms $(F1)$--$(F4)$ ensure that the  freshness symbol $\fresh$ is interpreted as a freshness relation.
Axioms $(A1)$, $(A2)$ define $\alpha$-equivalence of abstraction patterns as expected. 

We could further restrict the class of matching logic models considered by requiring that the carrier sets of sorts be nominal sets (i.e., elements have finite support), however,  this is not strictly necessary. There are also interesting models of nominal logic where the carrier sets are not finitely-supported nominal sets (elements could have infinite support as long as axiom $(F4)$ is satisfied, that is, the support of an element cannot include all the atoms); see~\citet{cheney:comhtn} 
and \citet{gabbay:genmn} for examples.

\subsubsection{Correctness of the representation of nominal logic}

Given an ordinary nominal logic theory $\Gamma$, we write $ML_{NL}(\Gamma)$ for the theory obtained by translating formulas in $\Gamma$ to matching logic patterns inductively and combining with the axioms of nominal logic as modified above. The translation of formulas is the same as translating FOL to matching logic (see~\citet{chen19lics}, section II.D). In the next propositions, provability for patterns is defined with respect to the complete Hilbert-style proof system for matching logic defined in previous work~\cite{rosu-2017-lmcs,chen19lics},  and provability of  nominal logic formulas with respect to the Hilbert-style proof system with Pitts' nominal logic axioms, which was shown complete in~\cite{cheney:comhtn}.

\begin{proposition}
\label{prop:axiom-translation}
The translations of $(E1)$--$(E5)$ are provable in matching logic from the axioms $Ax_{NLML}$.
\end{proposition}

\begin{proposition}\label{prop:translation}
There exists a translation from closed nominal logic formulas to NLML formulas, $\phi \mapsto \phi'$, such that $\phi': Pred$ and: 
\begin{enumerate}
    \item  If $\phi$ is provable in nominal logic then $\phi' = \top$ 
    in NLML.  

    \item If $\phi$ is not provable in nominal logic then  $\phi' = \bot$ 
    in NLML.  
\end{enumerate}
\end{proposition}

Although this approach to represent nominal logic in matching logic allows us to incorporate nominal features into matching logic with minimal effort, it restricts users to work within a specific theory and it does not exploit the full capabilities of matching logic to define partial and non-deterministic operators.  However, while axioms $(Fun_f)$ and $(Fun_c)$ force function symbols and constants imported from $NL$ to behave like functions, nothing stops us having non-deterministic or partial operations. We explore this possibility next.

\subsubsection{Partial and nondeterministic  operations}
\label{sec:sndapproach}
Matching logic supports function-like syntax for operations that are partial or multivalued (or both).  There is a natural elimination form for abstraction in nominal logic, called concretion~\cite{pitts:ns}, which takes an abstraction $t$ and a name $a$, and returns the body of the abstraction with the bound name replaced by $a$.  Concretion is only defined when the name is fresh for the abstraction.  Thus, in standard first-order logic concretion cannot be a function symbol, but in matching logic we can include  concretion as a partial operator.  
Likewise in $NL$ axiom $(F4)$  asserts that for any finite collection of values, there exists a name fresh for all of these values. In matching logic, we can directly define a freshness operator $fresh(-)$ which matches any name that is fresh for its argument.

The concretion operator is definable in NLML as follows, where  the $=$ symbol is coerced from $Pred$ to the sort of $y$ (recall that the interpretation of $\wedge$ is intersection).
\[
x @ a \defeq \exists y. y \wedge (\abs{a}{y} \coerce{=} x)\qquad (*)
\]
Let us walk through this definition to ensure its meaning is clear.  Informally, concretion instantiates the abstracted name of $x$ with $a$, but only if it is possible to safely rename to $a$.  The right-hand side encodes this meaning in matching logic.  First, we existentially quantify over some $y$ which will stand for the body of the abstraction with the bound name renamed to $a$.  The conjunction takes the intersection of the interpretations of the two patterns $y$ and $\abs{a}{y} \coerce{=} x$. The pattern $y$ simply matches $y$ itself, while $\abs{a}{y} \coerce{=} x$ either matches everything (if there is an $\alpha$-variant of $x$ where $a$ is the bound name) or nothing.  Thus we can read the pattern as ``$x @ a$ equals some $y$ provided $\abs{a}{y} = x$.''

Thus, Axioms (A1) and (A2) can be reformulated using abstraction as follows:
\[\begin{array}{cl}
\forall a,b:\nsort,x:\abs{\nsort}{\sort}. (\abs{a}{x})@b = (\swap{a}{b}{x} \wedge b \coerce{\fresh} \abs{a}{x}) & (A1')\\
\forall x:\abs{\nsort}{\sort}.x = \exists a. a \coerce{\fresh} x \wedge \abs{a}(x@a) & (A2')
\end{array}\]
Axiom $(A1')$ asserts a conditional form of beta-equivalence: if $b$ is  fresh for $\abs{a}{x}$ then $(\abs{a}{x})@b = \swap{a}{b}{x}$ (and in particular if $a= b$ then $(\abs{a}{x})@b = x$), otherwise if the freshness assumption does not hold the concretion pattern matches nothing.  Axiom $(A2')$ asserts eta-equivalence: any abstraction is equivalent to one where the bound name is fresh, and the body is concreted at that name.  (This axiom can be written more concisely using the $\new$-quantifier which we will consider later.) 
Note the use of a coercion operator ($\coerce{(-)}$) on $\fresh$, to transform $Pred$, the output of $\fresh$, into the sort $\abs{\nsort}{\sort}$ of $x$ so the intersection pattern $\wedge$ intersects sets of matching patterns at that sort.

These two axioms are equivalent to the classic nominal logic ones for abstraction:

\begin{proposition}\label{prop:abstraction}
Axioms $(A1')$ and $(A2')$ are valid in $NLML$ where concretion is defined using (*).  
Conversely, $(A1)$ and $(A2)$ follow from $(A1')$ and $(A2')$ in $NLML$ without $(A1),(A2)$.
\end{proposition}

The freshness operation $fresh_{\sort;\nsort}(-)$ is definable as follows:
\[
fresh_{\sort;\nsort}(x) \defeq \exists a:\nsort. a \wedge a \fresh^\dagger x \qquad (**)
\]
This operation matches any name that is fresh for $x$.  According to the semantics of patterns, $fresh(\phi)$, the pointwise extension of this operator, will therefore match any name that is fresh for something matching $\phi$; thus perhaps counterintuitively if $a,b$ are distinct names then $fresh(a \vee b) = \top$.  To see why, note that $a \vee b$ denotes a two-element set $\{a',b'\}$.  Every name is fresh for either $a'$ or $b'$ (or both), so all names are included in $\top$.  Indeed, it follows from the semantics of operations that for any unary operation $\sigma(\phi \vee \psi) = \sigma(\phi) \vee \sigma(\psi)$, so clearly $fresh(a\vee b) = fresh(a) \vee fresh(b) = \neg a \vee \neg b = \neg (a \wedge b) = \neg \bot = \top$.

This definition satisfies the following axioms which are reformulated versions of $(F1)$--$(F4)$.
\[
\begin{array}{cl}
\swap{fresh(x)}{fresh(x)}{x}= x & (F1')\\
fresh_{\nsort;\nsort}(a\colon\nsort) = \neg a & (F2')\\
fresh_{\nsort;\nsort'}(a:\nsort) = \top_{\nsort'} \qquad (\nsort \neq \nsort') & (F3')\\
\forall x:\sort. \exists a:\nsort. a \in fresh(x)&(F4')\\
\end{array}\]
\begin{proposition}\label{prop:fresh}
Axioms $(F1')$--$(F4')$ are valid in NLML where $fresh(-)$ is defined using $(**)$.  Conversely, axioms $(F1)$--$(F4)$ are provable from $(F1')$--$(F4')$.
\end{proposition}

\subsubsection{Assessment}
\label{sec:assessment}
We have now shown that nominal logic can be embedded as a first-order theory in matching logic, using the known embedding of FOL in matching logic, to obtain the combined system NLML.  We have even shown that matching logic offers some facilities that combine nicely with needs of nominal logic that are not well-addressed in a conventional FOL setting, such as catering for partial or nondeterministic operations like concretion and fresh name matching.  Are we done?  Why not stop here?  

We have three reasons for considering alternatives.  First, as noted in Section~\ref{sec:background}, in other nominal settings such as rewriting, logic programming and program verification, it has been found helpful to provide ground names instead of just variables of name sort (see~\cite{UrbanC:nomu-jv,cheney:simptn,FernandezM:nomr-jv,cheney:alppl, pitts11jfp}), and these are not available in NLML for the same reasons as in the original axiomatization of nominal logic.  Second, the $\new$-quantifier is a distinctive feature of nominal logic, and it is not clear whether, or how, it can be incorporated into matching logic in a first-class way, in particular, whether it can be made into a \emph{$\new$-pattern} that can be used anywhere in a pattern just like the other pattern connectives.  Third, and related to the first two points, as we shall discuss in Section~\ref{sec:examples}, typical reasoning in NLML usually requires considerable maintenance of freshness constraints, which can be mitigated using the $\new$-quantifier. 
Therefore, in the next section we consider a second system called NML (Nominal Matching Logic) that adopts ground names and a $\new$-quantifier as pattern constructs, where $\new$ binds names, and $\new$ can occur anywhere in a pattern, just like the other pattern constructs.

Note further that, as has been the case throughout the paper, we largely focus on syntax and semantics of the logics, and not on codifying their proof systems. For NLML this is no great loss since the established proof systems for matching logic can be used directly.  In the NML system presented in the next section this is not literally possible since we are adding to the syntax of patterns.  However as we shall see these additional patterns can be translated away if needed (although this has a cost in terms of freshness constraint maintenance and efficiency of unification algorithms); thus if desired the underlying proof system of matching logic can still be used.  We conjecture that a nicer proof system combining the techniques of \citet{chen19lics} and \citet{cheney16jlc} could be developed but leave this for future investigation.

\subsection{NML: Matching Logic with Names and $\new$}
\label{sec:newnml}
Finally we consider a further extension of matching logic in which syntax for names is introduced and the freshness quantifier ($\new$) is added as a general pattern construct and is given an interpretation directly in the semantics instead of being defined by an axiom scheme. We call this system Nominal Matching Logic (NML). 
Below we define the syntax and semantics of NML patterns and give some illustrative examples.

\subsubsection{NML Syntax}
A nominal matching logic signature $\mathbf{\Sigma}$  consists of $(S,\Var,Name, \Sigma)$ where
\begin{itemize}
\item
$\Sort$  is a non-empty set of sorts $\sort, \sort_1,\sort_2\ldots $, split into a set $\NSort$ of name sorts $\nsort, \nsort_1,\nsort_2,\ldots$, a set $\DSort$ of data sorts $\delta, \delta_1, \delta_2,\ldots$ including a sort $Pred$, and a set $\ASort$ of abstraction sorts $[\nsort]\sort$ --- there is one abstraction sort for each pair $\nsort,\sort$, 
\item 
$\Var$ is a $\Sort$-indexed $\{\Var_\sort\mid \sort\in \Sort\}$ of countable sets of variables $x\colon \sort,y\colon \sort,\ldots$, 
\item
$Name$ is an $\NSort$-indexed family $\{Name_{\nsort} \mid \nsort \in \NSort\}$ of countable sets of names $\Aa\colon \nsort,\Ab \colon \nsort,\ldots$ and 
\item
$\Sigma$ is a $(\Sort^* \times \Sort)$-indexed family  of sets of many-sorted symbols $\sigma$, written $\Sigma_{\sort_1,\ldots,\sort_n;\sort}$.

\end{itemize}

\begin{definition}[NML Syntax]
The syntax of  matching logic patterns is extended by including a distinguished category of names of name sort and a $\new$ pattern as follows:
\[\begin{array}{lcl}
\phi_\sort & ::= & x:\sort  \mid \Aa \colon \nsort  \mid \phi_\sort \wedge \psi_\sort \mid \neg \phi_\sort \mid \exists x{:}\sort'. \phi_\sort \mid  \\
&& \sigma(\phi_{\sort_1},\ldots,\phi_{\sort_n}) \mid \new \Aa {:} {\nsort}. \phi_{\sort}
\end{array}\]
where in the first case $x\in \Var_\sort$ and in the second $\Aa\in Name_{\nsort}$ where $\tau$ is a name sort $\alpha$, and in the case of $\sigma(\phi_{\tau_1},\ldots,\phi_{\tau_n})$ the operation symbol $\sigma$ must be in $\Sigma_{\tau_1,\ldots,\tau_n;\tau}$.  Both $\exists$ and $\new$ are binders (i.e., we work modulo $\alpha$-equivalence): variables that are not under the scope of a $\exists$ and  
names that are not under the scope of a $\new$  are said to be free. Substitution of variables by patterns avoids capture of free variables or free names.  Finally, names $\Aa$ are regarded as functional (single-valued) terms.
\end{definition}
Note that in $\new \Aa: \nsort. \phi_{\sort}$, $\nsort$ is required to be a name sort ($\new$ binds names rather than variables)
and names are now a separate syntactic class (in the semantics, names  behave like constants so that we always know that distinct free names occurring in a pattern are always different). 
In addition we assume that the set $\Sigma$  of symbols  includes the following families of symbols indexed by the relevant sorts (in the following we sometimes omit the subscripts when they are obvious from the context):
\[\begin{array}{rcll}
\swap{-}{-}{-} &:& \nsort \times \nsort \times \sort \to \sort    & \text{swapping (function)}\\
\abs{-}{-} & : & \nsort \times \sort \to \abs{\nsort}{\sort} & \text{abstraction (function)}\\
- @ - &:& \abs{\nsort}{\sort} \times \nsort \rightharpoonup \sort  & \text{concretion (partial function)}\\
fresh_{\sort,\nsort} &\in & \Sigma_{\sort;\nsort}   & \text{freshness (multivalued operation)} \\
- \fresh_{\nsort,\sort} - &:& \nsort \times \sort \rightharpoonup Pred  & \text{freshness relation (predicate)} \\
\coerce{-}  &\in&  \Sigma_{Pred;\tau} & \text{coercion operator, often implicit}
\end{array}\]

\subsubsection{Semantics}\label{sec:semantics}
We now formalize the semantics of nominal matching logic.  This semantics simultaneously generalizes the semantics of nominal logic~\cite{pitts:nomlfo-jv}  and of matching logic~\cite{chen19lics}, and takes into account the use of name constants~\cite{cheney:comhtn,UrbanC:nomu-jv}.
That is, an instance of nominal logic can be considered an instance of nominal matching logic in which all symbols correspond to function symbols or predicates, while an instance of matching logic can be considered an instance of nominal matching logic in which there are no name-sorts (and thus no abstraction sorts and no associated nominal pattern constructors).  

\begin{definition}[NML Model]
\label{def:NML-model}
Given an NML signature $\mathbf{\Sigma}$ with components $(S,\Var,Name, \Sigma)$, let $\mathbb{A}$ be $\bigcup_{\nsort\in \NSort} \mathbb{A}_{\nsort}$ where each  $\mathbb{A}_{\nsort}$ is  an infinite countable set of atoms and the $\mathbb{A}_{\nsort}$ are pairwise disjoint, and let $G$ be a product of permutation groups $\prod_i Sym(\mathbb{A}_i)$ (i.e., $G$ is the group of all sort-respecting permutations).
 A  nominal matching logic model $M = (\{M_{\sort}\}_{{\sort} \in {\Sort}}, \{\sigma_M\}_{\sigma \in \Sigma})$ consists of 
\begin{enumerate}
    \item a non-empty nominal $G$-set  $M_\sort$ for each $\sort \in \Sort$;
    \item an equivariant interpretation $$\sigma_M\colon M_{\sort_1} \times \cdots \times M_{\sort_n} \to \mathcal{P}_{fin}(M_\sort)$$ for each $\sigma \in \Sigma_{\sort_1,\ldots, \sort_n;\sort}$.
\end{enumerate}

We say a model is \emph{standard} if:
\begin{enumerate}
    \item the interpretation of each name sort $\nsort$ is the countably infinite set $\mathbb{A}_{\nsort}$;
    \item the interpretation of the sort $Pred$ is a singleton set $\{\star\}$, where $\star$ is equivariant, hence $\{\star\}$ is a nominal set whose powerset is isomorphic to Bool (as in standard models of matching logic, $\top$ is the full set and represents true and $\bot$ is the emptyset and represents false);
    \item the interpretation of each abstraction sort $\abs{\nsort}{\sort}$ is $\abs{M_{\nsort}}{M_\sort}$
    \item the interpretation of the swapping symbol $\swap{-}{-}{-} \colon \nsort \times \nsort \times \sort \to \sort$ is the swapping function on elements of $M_\sort$;
    \item the interpretation of the abstraction symbol is the quotienting function mapping $\langle a,x\rangle $ to its alpha-equivalence class, i.e. $\langle a,x\rangle  \mapsto \langle a,x\rangle /_{\equiv_\alpha}$ where $\equiv_\alpha$ is as defined in Section~\ref{sec:nl};
    \item the interpretation of the concretion symbol is the (partial) concretion function $(X,\Aa) \mapsto \{y\mid (\Aa,y) \in X\}$, more precisely, concretion applies to an abstraction $[a]v$ and a name $b$: if the name is fresh for the abstraction, it returns $\swap{a}{b}v$, otherwise it is undefined;
    \item the interpretation of the freshness operation $fresh_{\sort,\nsort}$ is the function that returns all the names in $\mathbb{A}_{\nsort}$ that are fresh for the argument, i.e., the  function $x \mapsto \{a \mid a \notin \supp(x)\}$;
    \item the interpretation of the freshness relation $\fresh_{\nsort,s}$ is the freshness predicate  on $\mathbb{A}_{\nsort}\times M_{\sort}$, i.e., it holds for the tuples $\{(a,x) \mid a \notin \supp(x)\}$.
\end{enumerate}
\end{definition}

As usual in matching logic, the pointwise extension 
$$\overline{\sigma_M} : \mathcal{P}_{fin}(M_{\sort_1}) \times \cdots \times \mathcal{P}_{fin}(M_{\sort_n}) \to \mathcal{P}_{fin}(M_\sort)$$ 
is defined as follows:
\[\begin{array}{rcl}
\overline{\sigma_M}(X_1,\ldots,X_n) &=& \bigcup\{\sigma_M(x_1,\ldots,x_n) \mid x_1 \in X_1,\ldots,x_n\in X_n\} \\
\end{array}\]

In particular,  $\psi_{\nsort} \fresh \phi_\sort$ is interpreted as $\top_{Pred}$ (that is, $M_{Pred} = \{\star\}$) if an instance of $\psi$ (a pattern of name sort) is fresh for an instance of $\phi$ (for example,  $\Aa \fresh \Aa \vee \Ab$ is interpreted as $\top$). Similarly,  $fresh(\phi_\sort)$  matches all the names that are fresh for some instance of $\phi_\sort$ (for example $fresh(\Aa \vee \Ab) = \bbA$).  On the other hand, note that $a = a \vee b$ does \emph{not} hold, because equality of patterns tests whether the two patterns have the same denotation.  Equality can therefore not be defined as a symbol, since the meaning of symbols is always determined by their behavior on individual values.  It would be possible to have an ``equality symbol'' $\sigma_=$ which, when restricted to single-element inputs, tests equality; however, for general patterns this operation would test whether the two argument patterns overlap, not whether they are equal.

Similarly, $\abs{\phi_{\nsort}}{\psi_\sort}$ is a pattern of sort $\abs{\nsort}{\sort}$ interpreted by the pointwise extension of the quotienting function associated with the abstraction symbol. For example, $\abs{\Aa \vee \Ab}{(\Aa \vee \Ab)}$ matches $\abs{\Aa}{\Aa}$, $\abs{\Aa}{\Ab}$, and $\abs{\Ab}{\Aa}$ (there are only three different values since $\abs{\Aa}{\Aa}= \abs{\Ab}{\Ab}$).  

\begin{definition}[Valuation]
\label{def:valuation}
A function $\rho : \Var \cup Name \rightharpoonup M$ with finite domain that is compatible with sorting (i.e., $\rho(x\colon\sort) \in M_\sort$, $\rho(\Aa\colon\nsort) \in M_{\nsort}$), injective on names and finitely supported is called a \emph{valuation}.  Injectivity ensures that two different names in the syntax are interpreted by different elements in the valuation. We say $\rho$ is a $\phi$-valuation when $dom(\rho) \supseteq FV(\phi) \cup FN(\phi)$, that is, all variables and names free in $\phi$ are assigned values by $\rho$. 
\end{definition}

Below, when we define the semantics of a pattern $\phi$ we implicitly assume that the valuation is a $\phi$-valuation.

Valuations, and more generally (finite-domain, partial) functions from $X$ to $Y$ where $X$ and $Y$ are nominal sets, can be seen as elements of the $G$-set $(Y_\bot)^X$ of all finitely-supported partial functions from $X$ to $Y$ (see~\cite{pitts11jfp} for details).   Therefore if $\pi$ is a well-sorted permutation in $\mathbb{A}$,  $\pi\cdot\rho$ is well defined: it is the valuation that maps $\Aa\in Name$ to $\pi\cdot \rho(\Aa)$ and $x\in \Var$ to $\pi\cdot \rho(x)$ (names in $Name$ and variables in $\Var$ have empty support and are not affected by permutations of atoms in $\mathbb{A}$). Note that permutations preserve the injectivity, sort-respecting, and finite support properties so the result of applying $\pi$ to $\rho$ is also a valuation.

Recall that a function $F$ is equivariant if $F(\pi\cdot e) = \pi\cdot F(e)$ for every $e$ in the domain of $F$ and every $\pi$. In particular, a valuation is equivariant if $(\pi\cdot \rho)(e) = \pi\cdot\rho( e)$ for every element in its domain.  

A valuation need not be equivariant but must be finitely supported, i.e. $\pi \cdot \rho = \rho$ whenever $supp(\pi) \cap \supp(\rho) = \emptyset$.
 The support of a valuation  is the union of the supports of the image of elements in its domain.

\begin{definition}[NML Pattern Semantics]
\label{def:NML-pattern-sem}
The meaning (set of matching elements) of a pattern $\phi$ for a given valuation $\rho$ 
is defined as shown in Figure~\ref{fig:semantics}.
\end{definition}
\begin{figure}[tb]
\begin{eqnarray*}
\semantics{\rho}{x:\sort} &=& \{\rho(x)\} \\ 
\semantics{\rho}{\Aa:\nsort} &=& \{\rho(\Aa)\} \\ 
\semantics{\rho}{\sigma(\phi_1,\ldots,\phi_n)} &=& \overline{\sigma_M}(\semantics{\rho}{\phi_1},\ldots,\semantics{\rho}{\phi_n})\\
\semantics{\rho}{\phi_1 \wedge \phi_2} &=& \semantics{\rho}{\phi_1} \cap \semantics{\rho}{\phi_2}\\
\semantics{\rho}{\neg\phi} &=& M_\sort - \semantics{\rho}{\phi}\\
\semantics{\rho}{\exists x:\sort. \phi}&=& \bigcup_{a\in M_\sort} \semantics{\rho[a/x]}{\phi}\\
  \semantics{\rho}{\new \Aa\colon \nsort. \phi} &=& \hspace{-1mm}\bigcup_{a \in \mathbb{A}_{\nsort}-supp(\rho)}\hspace{-.6cm}\{v\mid v \in \semantics{\rho[a/\Aa]}{\phi} \wedge a \not\in supp(v)\}
\end{eqnarray*}
\caption{Semantics of NML}\label{fig:semantics}
\end{figure}

A pattern $\new \Aa\colon \nsort. \phi_{\sort}$  matches those elements that match $\phi_\sort$ where $\Aa$ is instantiated with an atom $a\in \bbA_\nsort$ fresh for  $\rho$, and which do not have $a$ in their support. 
It could be written equivalently as:
\(\semantics{\rho}{\new \Aa\colon \nsort. \phi} = \{v\in \semantics{\rho[a/\Aa]}{\phi} \mid a \in  (\mathbb{A}_{\nsort}-supp(\rho)) - supp(v)\}
\).

Note that in the interpretation of the $\new$ pattern, the valuation $\rho$ is extended by assigning to $\Aa$ any element $a$ of $\mathbb{A}_{\nsort}$ that is fresh  (not in the support of the interpretations of free names and free variables  or instances of $\phi$). 
Compared to nominal logic, this differs because $\new$ is defined as a predicate symbol by an axiom scheme, whereas here the $\new$ quantifier can appear in an arbitrary place in a pattern.  
So for example $\exists x \colon \sort. \new \Aa\colon \nsort. \langle [\Aa]x,x\rangle$ is a pattern that characterizes pairs of abstractions and elements of $M_\sort$ where the abstracted name is fresh for the element.  Such a pattern has no direct equivalent in nominal logic.  This usage of $\new$ amounts to a form of local fresh name generation, sometimes denoted $\nu$ in other settings (e.g. Pitts' $\lambda\alpha\nu$-calculus~\citeyear{pitts11jfp}) to distinguish it from the fresh name quantifier occurring as a formula.  In NML there is no distinction between formulas and terms, so we use $\new$ in both places; also, $\new$ behaves differently than $\nu$ in $\lambda\alpha\nu$, for example $\nu a. a$ denotes an anonymous name rather than the empty set.

Before studying properties of NML we provide some simple examples. 

\begin{example}
\begin{itemize}
 \item Suppose $\phi$ does not contain $\Aa$ as a free name; then $\new \Aa.\phi $ is equivalent to $\phi$.
    
   This  result follows directly from the semantics of $\new$: if $\Aa$ is not free in $\phi$ then  $\semantics{\rho[a/\Aa]}{-}$ and $\semantics{\rho}{-}$ produce the same result. 
    \item $\phi_1 = \new \Aa. \Aa$ is a pattern that matches nothing (its interpretation is the empty set).
    Likewise  $\phi_2 = \new \Aa. \langle \Aa,\Aa\rangle$ 
    and $\phi_3 = \new \Aa. \new \Ab. \langle \Aa,\Ab \rangle$ are also empty.
    
    \item $\phi_4 = \new \Aa. \abs{\Aa}{\Aa}$ matches any abstraction whose body is the abstracted name. Note that $\new \Aa. \abs{\Aa}{\Aa}$ is a closed pattern, whereas $\abs{\Aa}{\Aa}$ has a free name $\Aa$ and can only be interpreted in a valuation that assigns a value to $\Aa$ (and then its denotation is the same as that of $\new \Aa. \abs{\Aa}{\Aa}$).
    
    \item $\phi_5 = \new \Aa. \Aa = \Aa$ is a valid predicate (equivalent to $\top$)
    
    \item $\phi_6 = \exists x. \new \Aa. \Aa = x$ is false/empty since whatever $x$ is, $\Aa$ must be chosen fresh for (and in particular distinct from) it.  However, $\new \Aa. \exists x. \Aa = x$ is true since we may choose $x= \Aa$.
    
    \item 
    $\phi_7 = \new \Aa. \abs{\Aa}{\phi}$ is a pattern that matches any abstraction where the abstracted atom is not in the support of (the interpretations of) free variables and free names of $\phi$ (aside from $\Aa$ itself) and the body is an instance of ${\phi}$ that does not contain the abstracted atom in its support.  
   
    To illustrate this kind of pattern, consider three possible rules representing eta-equivalence for the lambda-calculus:
    \[\begin{array}{rcl} 
    x\colon Exp &=& lam([\Aa]app(x,var(\Aa))) \\
    x\colon Exp &=& lam(\exists a. [a]app(x,var(a))) \\
    x\colon Exp &=& lam(\new \Aa. [\Aa]app(x,var(\Aa))) \end{array}\]
    In the first rule, only the specific name $\Aa$ can be used, and if it is also present in the support of $x$ then the result will be wrong (e.g. if $x = var(\Aa)$ then $\Aa$ is captured.)  The second rule allows eta-expanding using any name $a$, but still permits inadvertent capture of names in $x$.  The third rule correctly permits eta-expanding using any sufficiently fresh name $\Aa$ while ruling out variable capture.  The right-hand side of the third rule is also equivalent to $\exists a. a \fresh x \wedge lam([a]app(x,var(a)))$ and $\forall a. a \fresh x \Rightarrow lam([a]app(x,var(a)))$ (Prop.~\ref{prop:new-derived}).
    
    \end{itemize}
\end{example}

\subsubsection{Properties}\label{sec:properties}

In this section we show that the semantics is well defined in standard models, that is, the semantic interpretation $\rho \mapsto \semantics{\rho}{\phi}$ is an equivariant map from $\phi$-valuations to $M_\sort$ for every $\phi_\sort$ (Theorem~\ref{lem:equivariant-sem} and Cor.~\ref{cor:equiv}) and the expected equivalences of patterns are satisfied (Prop.~\ref{prop:abs-wedge-not}, \ref{prop:new-reordered} -   \ref{prop:new-abs-conc}).

\begin{definition}[Semantic equivalence of patterns in standard models]
    We say  $\phi$ and $\psi$ are equivalent patterns, written $\phi \iff \psi$, if for all suitable $\rho$ (whose domain includes the free variables and atoms of $\phi$ and $\psi$),  $\semantics{\rho}{\phi} = \semantics{\rho}{\psi}$.
     Alternatively, equivalence of patterns can be defined using equality: $\phi$ and $\psi$ are equivalent if for all suitable $\rho$ we have 
    $\semantics{\rho}{\phi = \psi} =  \top$.
\end{definition}

Recall that in standard models of NML, abstractions are interpreted as equivalence classes, therefore the alpha-equivalence relation generated by abstractions is already built into the semantics.  Moreover, the following properties of abstraction patterns follow directly from the definition of NML pattern semantics:

\begin{proposition}
\label{prop:abs-wedge-not}
\begin{enumerate}
    \item $[a](\phi_1 \wedge \phi_2) \iff [a]\phi_1 \wedge [a]\phi_2$
    \item $[a](\neg \phi) \iff \abs{a}{\top} \wedge \neg[a]\phi$ 
    
    \item \label{abs-conc}
    $(\abs{\Aa}{\phi}) @ \Ab \iff (\Ab \fresh \abs{\Aa}{\phi}) \wedge \swap{\Aa}{\Ab}{\phi}$
\end{enumerate}
\end{proposition}

\begin{theorem}[Equivariant Semantics]
\label{lem:equivariant-sem}
If $v \in \semantics{\rho}{\phi}$ then $\swap{a}{a'}{v} \in \semantics{\swap{a}{a'}{\rho}}{\phi}$.  I.e., for all $\phi$, $\swap{a}{a'}{\semantics{\rho}{\phi}}= \semantics{\swap{a}{a'}{\rho}}{\phi}$.
\end{theorem}

\begin{corollary}
\label{cor:equiv}
If $v \in \semantics{\rho}{\phi_\sort}$, $\rho(\Aa) = a$ and $\rho(\Aa')= a'$ then $\swap{a}{a'}{v} \in \semantics{\rho}{\swap{\Aa}{\Aa'}{\phi_\sort}}$ (because the swapping symbol is interpreted by the swapping function in the nominal set $M_\sort$).

In particular, for any pattern $\phi_\sort$, if  $a, a'$ are fresh for $v$ 
then $v\in \semantics{\rho}{\phi_\sort}$ if and only if $v\in \semantics{\rho}{\swap{\Aa}{\Aa'}\phi_\sort}$.
\end{corollary}

Since $\new \Aa.\phi$ patterns are matched by instances of $\phi$ obtained by assigning to $\Aa$ a fresh atom,
it is easy to see that if $\Aa$ is not free in $\phi$ then  $\new \Aa.\phi$ and $\phi$ have the same instances. Also, the order in which fresh atoms are quantified is not important (Proposition~\ref{prop:new-reordered}).

\begin{proposition}
\label{prop:new-reordered}
\begin{enumerate}
\item 
\label{new-empty}
If $\Aa$ is not free in $\phi$ then $\new \Aa.\phi\iff \phi$.
\item
$\new \Aa\colon \nsort. \new \Ab\colon \nsort'.\phi_\sort \iff \new \Ab\colon \nsort'. \new \Aa\colon \nsort.\phi_\sort$.
\end{enumerate}
\end{proposition}

\begin{proposition}
\label{prop:new-wedge-not}
The $\new$ pattern satisfies the following  properties:
\begin{enumerate}
\item 
\label{new-wedge}
\( \new \Aa\colon \nsort.\phi_\sort \wedge \new \Aa\colon \nsort.\psi_\sort
\iff 
\new \Aa\colon \nsort.(\phi_\sort \wedge \psi_\sort)
\).
In particular, $\new \Aa. (\phi ~\wedge ~\psi) \iff ((\new \Aa.\phi)~ \wedge ~\psi)$ 
if $\Aa$ is not free in $\psi$. 
A similar property holds for $\vee$ patterns:  \( \new \Aa\colon \nsort.\phi_\sort \vee \new \Aa\colon \nsort.\psi_\sort
\iff 
\new \Aa\colon \nsort.(\phi_\sort \vee \psi_\sort)
\).
\item
\( \neg \new \Aa\colon \nsort.\phi \iff \new \Aa\colon \nsort.\neg \phi  \)

\item 
 \( \abs{\Aa}{\new \Ab\colon \nsort.\phi_\sort} \iff \new \Ab\colon \nsort. \abs{\Aa}{\phi_\sort}\). 
 
 \item
 \label{new-sigma}
 $\langle \new \Aa. \phi, \psi \rangle \iff  \new \Aa. \langle \phi,\psi\rangle$ if 
$\Aa$ does not occur free in $\psi$. 

More generally, if 
 $\forall x_1,\ldots,x_n,  \Aa \# \sigma(x_1,...,x_n) \Rightarrow  \Aa \# x_1 \wedge \ldots \wedge \Aa \# x_n$ then
 $\sigma(\new \Aa. \phi_1,\ldots, \new \Aa. \phi_n) \iff  \new \Aa.\sigma(\phi_1,\ldots,\phi_n)$.
\end{enumerate}
\end{proposition}

The $\new$ pattern does not commute with abstraction for the same name: $\new \Aa.[\Aa]\Aa \neq [\Aa] \new \Aa.\Aa$ (the first one matches abstractions where the body is the abstracted name, whereas the second does not match anything). 
It does commute if the names are different:  $ \abs{\Aa}{\new \Ab.\phi}$ and $\new \Ab. \abs{\Aa}{\phi}$ are equivalent patterns that match an abstraction where the body refers to some fresh name $\Ab$ different from the abstracted name.  Note that when $\new$ occurs inside an abstraction the $\new$-bound name can always be renamed away from the abstracted term, this means that we can always hoist $\new$ out of the body of an abstraction.
However,  $\abs{ \new \Aa. \phi_{\nsort}}{\psi_\sort} \neq \new \Aa. \abs{\phi_{\nsort}}{\psi_\sort}$  even  if $\Aa$ is not free in $\psi$. For example,  $\abs{\new \Aa. \Aa}{\Ab}$ is matched by nothing, whereas $\new \Aa. \abs{\Aa}{\Ab}$ is matched by $\abs{a}{b}$ for two different $a,b$. Note that the abstraction symbol does not satisfy the condition in item~\ref{new-sigma} of Proposition~\ref{prop:new-wedge-not}.

The $\new$ pattern can be defined using existential (or universal) patterns and the freshness relation symbol as shown in Proposition~\ref{prop:new-derived}. Since the existential and universal patterns bind variables and $\new$ binds names, to define the $\new \Aa.\phi$ using $\exists$ or $\forall$ we need to replace the bound name $\Aa$ by a bound variable $z_\Aa$ which we assume does not occur anywhere else.  Finally, notice that in both existential and universal equivalent formulas, we introduce an existential variable $y$ standing for a matched result of the pattern, and the fresh name $z_\Aa$ must be fresh for this as well, corresponding to the constraint $a\notin supp(v)$ in the semantics of $\new$.  This is critical for the self-duality of $\new$ and to ensure that the universal and existential characterizations are equivalent.

\begin{proposition} 
\label{prop:new-derived}
The $\new$ pattern satisfies the following equivalences, 
where the notation $\phi_\sort(\Aa,\vec{\Ab},\vec{x})$ indicates that $\vec{x}$ are the free variables of $\phi$ and $\vec{\Ab}$  (and  possibly also $\Aa$) are the free names of $\phi$,  
$z_\Aa$ is a variable associated with the name $\Aa$, which we assume is not used anywhere else, and 
$\phi_\sort\{\Aa \mapsto z_{\Aa}\}$ is the result of replacing the  name $\Aa$ with $z_{\Aa}$ in $\phi_\sort$.
\begin{enumerate}
    \item 
$\new \Aa\colon \nsort. \phi_\sort(\Aa,\vec{\Ab},\vec{x})  \Leftrightarrow$\\
$\exists z_\Aa\colon \nsort. ((\exists y\colon \sort. y \wedge z_\Aa \coerce{\#}_s (\vec{\Ab},\vec{x},y)) \wedge \phi_\sort\{\Aa \mapsto z_{\Aa}\} 
)$

\item
$\new \Aa\colon \nsort. \phi_\sort(\Aa,\vec{\Ab},\vec{x}) \Leftrightarrow$\\
$\forall z_\Aa\colon \nsort. ((\exists y\colon \sort. y \wedge z_\Aa \coerce{\#}_s (\vec{\Ab},\vec{x},y)) \Rightarrow \phi_\sort\{\Aa \mapsto z_\Aa\} 
)$

\end{enumerate}

\end{proposition}

The following property links abstraction, concretion and $\new$, in a similar way as axiom $(A1')$ in NLML.
\begin{proposition} 
\label{prop:new-abs-conc}
$(\new \Aa. \abs{\Aa}{\phi}) @ \Ab  = \new \Aa. \swap{\Aa}{\Ab}{\phi}$.
\end{proposition}

We end this section with some observations relating NL, ML, NLML and NML.  The translation of the  axioms of nominal logic are valid in the standard models of NML: $(S1)$--$(S3)$ are satisfied since we interpret the swapping symbol with the swapping operation in the model, similarly $(E1)$--$(E5)$ are satisfied since the interpretation of  symbols $\sigma$ in the signature is equivariant, $(F1)$--$(F4)$ follow from the interpretation of the freshness operator and the interpretation of sorts as nominal sets (which ensures $(F4)$ holds). Axiom $(Q)$ follows from Prop.~\ref{prop:new-derived}. $(A1)$ and $(A2)$ are consequence of the interpretation of abstraction sorts as abstractions and the interpretation of the abstraction symbol.

Conversely, NML can be translated to NLML, by introducing (suitably freshness-constrained) variables to represent ground names and using Prop.~\ref{prop:new-derived} to eliminate $\new$-quantifiers.  Alternatively, NML can be translated to ordinary nominal logic with ground names (e.g. Cheney's system NLSeq~[\citeyear{cheney16jlc}]) following the same approach taken in translating ML to FOL.  Thus the proof systems already available for ML or NL can be leveraged to reason about NML.  A more appealing approach would be to develop a proof system for NML directly, perhaps buiding on NLSeq.  We leave this issue for future investigation.

\section{Examples}\label{sec:examples}

We now fix a specific model for reasoning about the typed lambda-calculus, with sorts including $Exp$ (expressions), $Ty$ (types), and $Var$ (variables, a name-sort).  These sorts are interpreted as nominal sets $M_{Var}$, $M_{Exp}$, and $M_{Ty}$ satisfying the following equations:
\[\begin{array}{rcl}
    M_{Exp} &=& M_{Var} + (M_{Exp} \times M_{Exp}) + \abs{M_{Var}}{M_{Exp}} \\
    M_{Ty} &=& 1 + M_{Ty} \times M_{Ty} + \cdots \end{array}\]
The three cases for expressions correspond to $var$, $app$ and $lam$ as discussed in Section~\ref{sec:background}.  We assume at least one constant type (e.g. $int$ or $unit$) and assume there is a binary constructor $fn : Ty \times Ty \to Ty$ for function types but leave the exact set of types unspecified otherwise.
It is well-established that the abstraction construction satisfies the required properties (monotonicity, continuity) so that least fixed points of equations involving abstractions exist according to nominal versions of standard fixed-point theorems~\cite{pitts:alpsri}, and the resulting models are finitely-supported. Therefore the initial algebra $M_{Exp}$ is essentially the set of lambda-terms quotiented by alpha-equivalence.  We fix $M_\Lambda$ as the standard model obtained taking $M_{Exp}$ and $M_{Ty}$ as defined above, and in the rest of this section statements are relative to this model.

As indicated at the end of the previous section,  neither NLML nor NML have been implemented as part of an automated system.  The examples and proofs in this section are carried out by hand and have not been mechanically checked.

\subsection{The Lambda Calculus in NLML}\label{sec:lam-nlml}
Recall the signature for lambda-terms and axioms for substitution from Section~\ref{sec:background}, see Example~\ref{ex:lamc}.
Consider the induction principle, which is schematic over a predicate $P\in \Sigma_{ Exp, \sort_1,\ldots,\sort_n;Pred}$:
\begin{arcl}
(\forall a\colon Var. P(var(a),\vec{y})) &\Rightarrow &\\
(\forall t_1\colon Exp,t_2\colon Exp. P(t_1,\vec{y}) \wedge P(t_2,\vec{y}) \Rightarrow P(app(t_1,t_2),\vec{y})) &\Rightarrow& \\
  (\forall a, t.  a \fresh \vec{y} \Rightarrow P(t,\vec{y}) \Rightarrow P(lam(\abs{a}{t})),\vec{y}) &\Rightarrow& \\
  \forall t\colon Exp. P(t,\vec{y})&&
\end{arcl}
This induction principle deserves some explanation.  We make any other parameters of the induction hypothesis (e.g. variables with respect to which any bound names must be chosen fresh) explicit via a parameter list $\vec{y}$.  The cases for variables and application are otherwise standard.  In the case for lambda, the induction step we need to show is that if $P(t,\vec{y})$ holds for some \emph{sufficiently fresh} name $a$ (i.e. for $a \fresh \vec{y})$ then $P(lam(\abs{a}{t}),\vec{y})$ holds.  For more details, see \citet{pitts:alpsri} and \citet{UrbanC:ntih}.

Now suppose we wish to prove the following (standard) property of substitution:
\begin{arcl}
P(x,y,z,y',z')  &\defeq&  y \fresh y',z' \Rightarrow subst(subst(x,y,z),y',z') \\ && = subst(subst(x,y',z'),y,subst(z,y',z')).
\end{arcl}
We proceed by induction on $x$.  The base case where $x$ is a variable is straightforward by case analysis on $y$.  The case where $x$ is an application is also straightforward. 

For the case of lambda-abstraction, let $a,t$ be given such that $a \fresh y,z,y',z'$.  We need to show:
\[
subst(subst(x,y,z),y',z') = subst(subst(x,y',z'),y,subst(z,y',z') )\\
\]
implies \[
\begin{array}{cl}
& subst(subst(lam([a]x),y,z),y',z') \\
= & subst(subst(lam([a]x),y',z'),y,subst(z,y',z')).  
\end{array}
\]

Accordingly, reason as follows:
\begin{arcl}
  &&subst(subst(lam([a]x),y,z),y',z')\\
  &=& subst(lam([a](subst(x,y,z)),y',z')\\
&=&  lam([a]subst(subst(x,y,z),y',z'))
\end{arcl}
using the substitution axiom $(Subst4)$ (see Example~\ref{ex:lamc}) twice since we know $a \fresh y,z,y',z'$.

We can now apply the induction hypothesis, relying again on the freshness assumptions for $a$:
\begin{arcl}
&& lam([a]subst(subst(x,y,z),y',z')) \\
&=& 
lam([a]subst(subst(x,y',z'),y,subst(z,y',z')))
\end{arcl}
Finally, we apply again the substitution axioms twice:
\begin{arcl}
&& lam([a]subst(subst(x,y',z'),y,subst(z,y',z')))\\
&=& subst(lam([a]subst(x,y',z')),y,subst(z,y',z'))\\
&=&
subst(subst(lam([a]x,y',z')),y,subst(z,y',z'))
\end{arcl}
The final two steps require  the freshness conditions:  \\ $a \fresh y',z',y,subst(z,y',z')$.  
By assumption, $a \fresh y',z',y$ holds. 
The final freshness assertion requires an additional lemma:
$$a \fresh x,y,z \Rightarrow a \fresh subst(x,y,z)$$ 
which holds by equivariance of  symbols.  

Nominal logic supports proofs by induction for syntax with binding operators. NLML inherits the inductive reasoning techniques available in nominal logic,  as illustrated in this example. Since NLML is an instance of matching logic with a specific theory, it can be used in any implementation of matching logic. However, as discussed in Section~\ref{sec:assessment}, 
there is an alternative, more concise way to specify the syntax of the lambda calculus and the induction principle using $\new$, which can be written in NML as shown below. 

\subsection{The Lambda Calculus in NML}\label{sec:lam-nml}

Using NML we can axiomatize the substitution operation in the $\lambda$-calculus as follows:
\begin{arcl}
subst(var(a),a,z) &=& z\\
 subst(var(a),\neg a,z) &=& var(a)\\
subst(app(x_1,x_2),y,z) &=& app(subst(x_1,y,z),subst(x_2,y,z))\\
subst(lam(x),y,z) & = & lam(\new \Aa.\abs{\Aa}{subst(x@\Aa,y,z)})
\end{arcl}
Notice that this is now a completely equational definition, no side conditions are required (cf. the axiomatization of substitution given in Example~\ref{ex:lamc}). 
The last axiom can be  replaced with the equivalent:
$$subst(lam(x),y,z) = \new \Aa. lam(\abs{\Aa}{subst(x@\Aa,y,z)})$$
 (see Proposition~\ref{prop:new-wedge-not}, part~\ref{new-sigma}, which applies since $lam$ acts as a constructor).

The induction principle for expressions can be formulated using $\new$ and concretion to deal with the case of lambda-abstractions, avoiding the need for freshness constraints (compare with the version of the induction principle in the previous section):
\begin{arcl}
  (\forall x\colon Var. P(var(x))) &\Rightarrow& \\
(\forall t_1:Exp,t_2:Exp. P(t_1) \wedge P(t_2) \Rightarrow P(app(t_1,t_2))) &\Rightarrow& \\
(\forall t:\abs{Var}{Exp}.   \new \Aa:Var. P(t@\Aa) \Rightarrow P(lam(t)) &\Rightarrow&\\
\forall t:Exp. P(t) &&
\end{arcl}
Note that because we use the $\new$-quantifier in the case for lambda, the freshness side-conditions $a \fresh \vec{y}$ are implicit, and so we can omit the additional parameters $\vec{y}$, which makes this induction principle simpler and cleaner. In the lambda case, we need to prove that if $P(t@\Aa)$ holds for a fresh name $\Aa$ then $P(lam(t))$ holds; it is somewhat easier to see how to generalize this principle systematically to any nominal datatype.

Also note that thanks to the distinction between names and variables, the Substitution Lemma can now be stated as shown below with just one freshness condition, which formalizes the usual side-condition in textbook statements~\cite{BarendregtH:lamcss} (we do not need to specify $\Aa \neq \Ab$ explicitly since this always holds for two different atoms).  
\[\begin{array}{l}
\Aa \fresh z' \Rightarrow \\
subst(subst(x,\Aa,z),\Ab,z') = subst(subst(x,\Ab,z'),\Aa,subst(z,\Ab,z'))
\end{array}\]

This means that the induction step in the lambda case becomes:
\[
\begin{array}{c}
(\new \Ac. \Aa \fresh z' \Rightarrow subst(subst(x@\Ac,\Aa,z),\Ab,z') \\
\qquad\qquad\qquad = subst(subst(x@\Ac,\Ab,z'),\Aa,subst(z,\Ab,z')) 
\Rightarrow \\
(\Aa \fresh z' \Rightarrow subst(subst(lam(x),\Aa,z),\Ab,z') \\
\qquad \qquad\qquad= subst(subst(lam(x),\Ab,z'),\Aa,subst(z,\Ab,z')))
\end{array}
\]
The reasoning for the lambda case is as follows using the fourth axiom for substitution twice:
\begin{arcl}
&& subst(subst(lam(x),\Aa,z),\Ab,z') \\&=&
subst(lam(\new \Ac. \abs{\Ac}{subst(x@\Ac,\Aa,z)}),\Ab,z')\\
 &=&
lam(\new \Ad.\abs{\Ad}{subst((\new \Ac. \abs{\Ac}{subst(x@\Ac,\Aa,z)})@\Ad,\Ab,z')}))
\end{arcl}
By Prop.~\ref{prop:new-abs-conc}, $(\new \Ac. \abs{\Ac}{subst(x@\Ac,\Aa,z)})@\Ad = \new \Ac.\swap{\Ac}{\Ad}{subst(x@\Ac,\Aa,z)}$. The latter is equivalent to $\new \Ac.subst(x@\Ad,\Aa,z)$,
since $\Ac$, $\Ad$ are new (not in the support of the variables in the pattern).
By Prop.~\ref{prop:new-reordered}(\ref{new-empty}), the latter is equivalent to
$subst(x@\Ad,\Aa,z)$.
Therefore, 
\begin{arcl}
    && subst(subst(lam(x),\Aa,z),\Ab,z')\\
  &=& lam(\new \Ad.\abs{\Ad}{subst(subst(x@\Ad,\Aa,z), \Ab,z')}))
\end{arcl}
By induction hypothesis,  the above equals
$$lam(\new \Ad. \abs{\Ad}{subst(subst(x@\Ad,\Ab,z')},\Aa,subst(z,\Ab,z')))\,.$$
Applying again Prop.~\ref{prop:new-abs-conc} and Prop.~\ref{prop:new-reordered}(\ref{new-empty}), the latter equals to 
$$lam(\new \Ad. \abs{\Ad}{subst((\new \Ac. \abs{\Ac}{subst(x@\Ac,\Ab,z')})@\Ad,\Aa,subst(z,\Ab,z'))}).$$ 
We can now use $(Subst4)$ to complete the proof:
\begin{arcl}
  &&subst(subst(lam(x),\Aa,z),\Ab,z')\\
  &=& subst(lam(\new \Ac. \abs{\Ac}{subst(x@\Ac,\Ab,z')}),\Aa,subst(z,\Ab,z'))\\
&=& subst(subst(lam(x,\Ab,z')),\Aa,subst(z,\Ab,z'))
\end{arcl}

\subsection{Reduction and Well-Formedness}
Using the same signature for $\lambda$-terms as in the previous section, 
reduction $red \in \Sigma_{Exp;Exp}$ can be defined in NML as follows:
\[\begin{array}{rcl}
    red(app(x,y)) &=& app(red(x),y) \vee app(x,red(y))\\
    &\vee& (\exists z. x = lam(z) \wedge \new \Aa.subst(z@\Aa,\Aa,y))\\
red(lam(\abs{\Aa}{y})) &=& lam(\abs{\Aa}{red(y))}
\end{array}\]
These can be viewed as first-order axioms or as an inductive definition using a least fixed point pattern $\mu$, following the strategy for defining operations inductively presented by Chen and Rosu for matching $\mu$ logic~\cite{chen19lics}. The first axiom defines reduction for applications: it states that the instances of $red(app(x,y))$ are obtained by reducing in the first argument, or in the second argument, or applying the $\beta$-rule at the root if the first argument is a $\lambda$-abstraction, using the capture-avoiding substitution symbol $subst$ previously defined.
These are the only instances of $red$.

Type checking can also be characterized in NML as shown below, 
using sorts  $Ty$ for types and $Map[Var, Ty]$ for finite maps from variables to types  (we abbreviate $c \colon Map[Var, Ty]$ as  $c\colon Ctx$). The symbol $wf \in \Sigma_{Ctx, Ty; Exp}$, denoting well-formed expressions of type $Ty$ in the typing context $Ctx$, satisfies the following axiom:
\[\begin{array}{rcll}
wf(c,t) &=& \exists a:Var. var(a) \wedge t = c[a]\\
&\vee& \exists u:Ty. app(wf(c,fn(u,t)), wf(c,u))\\
&\vee& \exists t_1,t_2: Ty.  t = fn(t_1,t_2) \wedge\\
&& \hspace{1.5cm} lam(\new \Aa. \abs{\Aa}wf(c[\Aa:=t_1],t_2))
\end{array}\]
Here we write $c[a]$ for the (partial) operation that looks up $a$'s binding in $c$ and write $c[a:=t]$ for the (partial) operation that extends a finite map $c$ with a binding for a variable not already present in its domain.
Notice that $wf(c,t)$ is a pattern which matches just those lambda-terms that are well-formed in context $c$ and have type $t$: The first line in the definition of $wf(c,t)$ indicates that a variable $var(a)$ is a well-formed term of type $t$ in $c$ if $c[a] = t$; the second line specifies the usual typing rule for applications: the first argument must have arrow type $fn(u,t)$ in the context $c$, while the second argument must have type $u$ in $c$; the third and fourth lines specify the typing rule for $\lambda$ expressions: a well-formed $\lambda$-abstraction of type $fn(t_1,t_2)$ in $c$ is an instance of the pattern $lam(\new \Aa. \abs{\Aa}wf(c[\Aa:=t_1],t_2))$. For example,  $lam(\abs{\Aa}{var(\Aa)}) \in wf([],fn(t,t))$ holds for any type $t$.

Finite maps satisfy standard axioms, such as
\[\begin{array}{c}
c[a:=t][a] = t \qquad a \neq b \Rightarrow c[a:= t][b] = c[b] \\ x \in c[a] \wedge y \in c[a] \Rightarrow x = y
\end{array}\]
and we abuse notation slightly by writing $c \subseteq c'$ for $\forall a. c[a] \subseteq c'[a]$.

\paragraph{Subject Reduction}
The standard property of subject reduction can be stated as follows:
\[red(wf(c,t)) \subseteq wf(c,t)\]
that is, all the reducts of a well-formed term of type $t$ in context $c$ (i.e., the instances of $red(wf(c,t))$) are well-formed terms of type $t$ in context $c$. In other words, given a well-formed term  of type $t$ in context $c$, reducing it yields another well-formed term of the same type.  
This property can be proved  using the axioms for $red$ and $wf$ and the following induction principle:
\[\begin{array}{rcl}
(\forall t\colon Ty, a\colon Var, c\colon Ctx . (var(a) \wedge t = c[a]) \subseteq P(c,t)) &\Rightarrow& \\
(\forall t_1,t_2\colon Ty, c\colon Ctx . app(P(c,fn(t_1,t_2)),P(c,t_1)) \subseteq  P(c,t_2))  &\Rightarrow& \\
(\forall  t_1,t_2\colon Ty, c\colon Ctx. lam( \new \Aa:Var. \abs{\Aa}{P(c[\Aa:=t_1],t_2)}) &&\\
\subseteq P(c,fn(t_1,t_2)))    &\Rightarrow&\\
\forall c\colon Ctx, t\colon Ty. wf(c,t) \subseteq P(c,t)
\end{array}\]
together with a lemma stating that well-typed substitutions preserve types:
\[\begin{array}{rl}
\forall a, t, t', c, c'. & c' \subseteq c \\ 
& \Rightarrow subst(wf(c[a:=t'],t),a,wf(c',t’)) \subseteq wf(c,t)
\end{array}\]

First we show how to prove this lemma using the above-mentioned induction principle and the axioms for $subst$ and $wf$.

We start by using the definition of $wf$ and the semantics of symbols, which extends pointwise to sets of arguments, allowing us to move  the $subst$ symbol inside $\exists$ and $\vee$ patterns:
\[\begin{array}{lcl}
    &&subst(wf(c[a:=t'],t),a,wf(c',t))
    \\
    &=& subst(\exists a':Var. var(a') \wedge t = (c[a:=t'])[a']\\
&& \vee\exists u:Ty. app(wf(c[a:=t'],fn(u,t)), wf(c[a:=t'],u))\\
&& \vee \exists t_1,t_2: Ty.  t = fn(t_1,t_2) \wedge \\
&& \hspace{.2cm}lam(\new \Ab. \abs{\Ab}wf(c[a:=t'][\Ab:=t_1],t_2)),a, wf(c',t’))\\
& = & \exists a':Var. subst(var(a') \wedge t = (c[a:=t'])[a'],a, wf(c',t’))\\
&& \vee\exists u:Ty. subst(app(wf(c[a:=t'],fn(u,t)),wf(c,u)),\\
&& \hspace{5.5cm} a, wf(c',t’))\\
&& \vee \exists t_1,t_2: Ty.  subst(t = fn(t_1,t_2) \wedge \\
&& \hspace{.5cm}lam(\new \Ab. \abs{\Ab}wf(c[a:=t'][\Ab:=t_1],t_2)), a, wf(c't'))\\
\end{array}\]
Note that $(c[a:=t'])[a] = t'$ and to prove the lemma it is sufficient to prove that each of the disjuncts is included in $wf(c,t)$. 

For the first disjunct, we consider whether the variable $a'$ is equal to $a$ or not.  If so, then $t = c[a:=t'][a] = t'$ and by $(Subst1)$ (see Example~\ref{ex:lamc}) $subst(var(a),a, wf(c',t')) \subseteq wf(c',t') \subseteq wf(c,t)$ as desired since $c' \subseteq c$ and $t = t'$.  Otherwise, $(c[a:=t'])[a'] = c[a']$ and $subst(var(a'),a,wf(c',t')) = var(a')$ and by definition of $wf$ we have $\exists a'. var(a') \wedge t = c[a'] \subseteq wf(c,t)$.

In the second disjunct, using $(Subst3)$ to move $subst$ inside $app$, the induction hypothesis, and the definition of $wf$, we obtain the required inclusion. 

Finally, in the third disjunct again we can move $\new$ out (Prop.~\ref{prop:new-wedge-not}(\ref{new-sigma})) and then apply axiom $(Subst4)$ to move $subst$ inside $lam$, completing the proof by using the induction hypothesis again when $t = fn(t_1,t_2)$ (otherwise the disjunct is empty and the inclusion holds trivially). 

To prove Subject Reduction, again we start by using the definition of $wf$:
\[\begin{array}{lcl}
    && red(wf(c,t)) \\
    &=& red(\exists a:Var. var(a) \wedge t = c[a]\\
&& \vee\exists u:Ty. app(wf(c,fn(u,t)), wf(c,u))\\
&& \vee \exists t_1,t_2: Ty.  t = fn(t_1,t_2) \wedge lam(\new \Aa. \abs{\Aa}wf(c[\Aa:=t_1],t_2))) \\
& = & \exists a:Var. red(var(a) \wedge t = c[a])\\
&& \vee\exists u:Ty. red(app(wf(c,fn(u,t)),wf(c,u))\\
&& \vee \exists t_1,t_2: Ty.  red(t = fn(t_1,t_2) \wedge \\
&& \hspace{3.5cm}lam(\new \Aa. \abs{\Aa}wf(c[\Aa:=t_1],t_2)))\\
\end{array}\]

Note that $red(var(a)) = \bot$ so the first disjunct is empty. 

In the second disjunct we can use the first axiom for $red$. Notice that by Prop.~\ref{prop:new-abs-conc}, the following is equivalent to the original axiom involving concretion: $red(app(x,y)) = app(red(x),y) \vee app(x,red(y)) \vee (\new \Aa.\exists z.  x = lam([\Aa]z) \wedge subst(z,\Aa,y))$.

In the third disjunct we can use the second axiom for $red$ after moving the $\new$ out (Prop.~\ref{prop:new-wedge-not}(\ref{new-sigma})). Hence,
\[\begin{array}{lcl}
&&red(wf(c,t))\\
&=&  \exists u\colon Ty. app(red(wf(c,fn(u,t))),wf(c,u)) \\
&& \hspace{1 cm} \vee  app(wf(c,fn(u,t)),red(wf(c,u)))  \\

&& \hspace{1cm} 
\vee (\new \Aa.\exists z.  wf(c,fn(u,t))  \supseteq lam([\Aa]z)\\
&&\qquad\qquad\wedge subst(z,\Aa,wf(c,u)))\\

&& \vee \exists t_1,t_2: Ty.  t = fn(t_1,t_2) \\
&&\qquad\qquad\wedge \new \Aa.lam( \abs{\Aa}red(wf(c[\Aa:=t_1],t_2)))\\
\end{array}\]

We can push the existential quantifiers above under disjunction to get four disjuncts.  It remains to prove that each disjunct above is included in $wf(c,t)$. The first two and the last inclusions follow directly by induction. 
The third follows from the lemma above, noticing that  
$$ wf(c,fn(u,t)) \supseteq lam([\Aa]z) \Rightarrow  z \subseteq wf(c[\Aa:=u],t) $$ 
by definition of $wf$.

\paragraph{Progress}
Another standard property in a syntactic proof of type soundness  states that a well-formed closed term that is not weakly-reducible is a value, where in this case $value = lam(\top)$ since we consider only a pure lambda-calculus.  
This is equivalent to saying that a well-formed closed term is either a value or can be reduced, so can be written as follows:
\[
wf([],t) \subseteq value \vee reducible
\]
where $reducible$ is defined as $\exists x. x \wedge \exists y. y \in red(x)$, that is, $reducible$ matches those terms that can reduce in some way.
We prove some auxiliary properties of $reducible$:  
\[ app(reducible,\top) \subseteq reducible \qquad app(lam(\top),\top) \subseteq reducible\]
To show $app(reducible,\top) \subseteq reducible$ suppose we have some matches $app(reducible,\top)$, that is, it is an application $app(x,y)$ whose first argument is reducible, i.e. there is some $x' \in red(x)$.  Hence $app(x',y)$ matches $red(app(x,y))$ so $app(x,y)$ is reducible.  To show $app(lam(\top),\top) \subseteq reducible$, again suppose $app(lam(x),y)$ matches $app(lam(\top),\top)$.  Then 
$$\exists z.lam(x)  = lam(z) \wedge \new \Aa. subst(z @ \Aa, \Aa,y) \subseteq red(app(lam(x),y))$$ 
so we just need to show the contained pattern is inhabited.  Clearly choosing $x = z$ satisfies the first conjunct and it remains to show that the substitution pattern is defined.  This can be proved by induction on the first argument to $subst$.

We now outline a proof sketch for the progress property, which we reformulate slightly to $c = [] \Rightarrow wf(c,t) \subseteq value \vee reducible$.
Assume $c = []$.  Proceed by induction on $wf(c,t)$ to prove $wf(c,t) \Rightarrow value \vee reducible$.  
\begin{itemize}
\item Case:  $\exists a. var(a) \wedge t = c[a] \subseteq value \vee reducible$:  $c= []$ so $c[a]$ is empty and so $\exists a. var(a) \wedge t = c[a] = \bot \subseteq value \vee reducible$.

\item Case: $app(value \vee reducible, value \vee reducible) \subseteq value \vee reducible$.  We consider two cases.  If the first argument is reducible then, by definition of $reducible$: \\
$app(reducible, value \vee reducible) \subseteq app(reducible, \top)\\
\subseteq value \vee reducible$.  \\
Otherwise,  the first argument is a value and since $value = lam(\top)$ we have \\
\begin{eqnarray*}
app(lam(\top), value \vee reducible) &\subseteq &
app(lam(\top), \top)\\
&\subseteq& value \vee reducible
\end{eqnarray*}
again by definition.  (Here the reasoning is simpler because in the pure simply typed lambda calculus there is only one form of value; for a larger language with some other value forms we would need a separate lemma ensuring that values of function type are lambdas.)

\item Case: $lam(\new a. \abs{a}{value \vee reducible}) \subseteq value \vee reducible$.  This follows immediately since $lam(\new a. \abs{a}{value \vee reducible}) \subseteq lam(\top) = value$.
\end{itemize}

Although in many cases occurrences of $\new$ are adjacent to abstraction (which could justify having  syntactic sugar for this), the examples in this section, in particular the induction principles, illustrate the advantage of keeping  $\new$ independent of abstraction.

\section{Comparison with binding in Applicative Matching Logic}\label{sec:comparison}

Binders, as exemplified by $\lambda$ in the $\lambda$-calculus, can be encoded in 
matching logic by using an existential pattern, a special symbol \textsf{intension} and  a sort containing representatives of all possible function graphs, as shown by \citet{chen20icfp}. The existential pattern represents a function as a set of pairs (argument, value) and \textsf{intension} packs this set as an object in a power sort.  For example, to represent the $\lambda$-term $\lambda x.e$, we apply \textsf{intension} to $\exists x\colon Var.\langle x,e\rangle$, which produces an element of the power sort $2^{Var\times Exp}$, and introduce a symbol $\lam \in \Sigma_{2^{Var\times Exp};Exp}$ that decodes it into the intended interpretation of $\lambda x.e$. Abbreviating \textsf{intension} $\exists x\colon Var.\langle x,e\rangle$  as $[x\colon Var] e$, we obtain a notation like that for abstraction in nominal logic, but note that whereas the abstraction $[x]e$ in nominal logic denotes an $\alpha$-equivalence class of pairs, in matching logic it denotes the combination of an existential pattern (matching a union of pairs) with an intension operator that packs the union into an object --- essentially a second-order construction. 

Given an expression $\lam([x]e)$ and a fresh variable $y$, in nominal logic the $\alpha$-equivalence $\lam([x]e) = \lam([y]\swap{x}{y}{e})$ holds by construction (given the axiomatisation of swapping and abstraction), whereas in matching logic only the $\alpha$-equivalence between the existential patterns $\exists x\colon Var.\langle x,e\rangle$ and $\exists y\colon Var.\langle x,e\{x \mapsto y\}\rangle$ holds by construction.  We can then prove the alpha-equivalence $[x\colon Var] e = [y\colon Var](e\{x \mapsto y\})$  by relying on the semantics of  existential patterns, and the semantics of the \textsf{intension} symbol.

Once we have a mechanism to define syntax with binders, we can define the $\beta$-reduction relation. In both nominal and matching logic, we can introduce axioms to specify this relation. For example, \citet{chen20icfp} give the following $(\beta)$ axiom schema:
\[\forall x_1\colon Var.\cdots \forall x_n\colon Var.\ \app(\lam [x] e , e') = e[e'/x]\]
where $e, e'$ are patterns encoding  two $\lambda$-calculus expressions, $x_1,\ldots,x_n$ are the free variables of $(\lambda x. e)~e'$,  and $e[e'/x]$ denotes the  meta-level substitution of matching logic.  In other words, we have one axiom for each possible $\beta$-redex, for example:
$$\forall x y z.  (\lam [w]\app(\app(w, w), \app(y, z)) = \app(\app(y,z) , \app(y,z))$$

In contrast, in nominal logic (or NLML) one can write this as a single (conditional) equational law
\[\forall e, e', a.\ a \# e' \Rightarrow \app(\lam (\abs{a}{e}),e') = \subst(e,a,e')\]
where $\subst$ is a function symbol denoting explicit substitution (defined via a set of equations, see Section~\ref{sec:nl}). 
Alternatively using quantifier alternation we can   specify the $\beta$-axiom as follows:
\[\forall e'. \new \Aa. \forall e. \app(\lam([\Aa]e)),e') = \subst(e,\Aa,e')\;,\]
or, in NML, even as an equational axiom using $\new$ and $@$:
\[\forall e':Exp.  \forall e:\abs{Var}{Exp}. \app(\lam(e),e') = \new \Aa. \subst(e@\Aa,\Aa,e')\;.\]

\section{Additional Related Work}\label{sec:related}

Matching logic~\cite{rosu-2017-lmcs} is the foundation of the $\bbK$ semantic framework, which has been used to specify programming languages (e.g., Java and C), virtual machines (e.g. KEVM~\cite{csf}), amongst others. Variants of matching logic have been defined to include recursive pattern definitions (see Matching $\mu$-logic~\cite{chen19lics}) and facilitate automated reasoning~\cite{chen20oopsla}. To specify programming languages that include binding operators, matching logic uses an encoding 
based on the internalisation of the graph of a function from bound name to expressions containing bound names~\cite{chen20icfp}. 
This is sufficient to formalise the semantics of binding operators, but it is unclear how to reason by induction on the defined syntax. 

Another approach to encode binders in matching logic is to use numerical codes to represent bound variables, as in de Bruijn encodings of the $\lambda$-calculus. This approach works well from the computational point of view, but reasoning on de Bruijn's syntax is generally considered less intuitive than on syntax with names.  On the other hand, reasoning about ``raw'' syntax without any support for alpha-equivalence is also painful because operations that intuitively seem like functions, such as substitution, need to be treated with care.  Because intuitive reasoning about names and binding is challenging, over several decades many other techniques for representing and reasoning about binding have been investigated, including higher-order abstract syntax~\cite{pfenning:hoas}, and locally-nameless encodings~\cite{PollackR:metatheory}, to name just a few.  Could they be combined with matching logic instead?  For de Bruijn encodings the answer is certainly yes, but this would do little to overcome the gap between on-paper and formal reasoning.  For other approaches, particularly higher-order abstract syntax, where the typed lambda-calculus is used as a metalanguage and its binding structure used to encode bindings in object languages, the answer is less clear: induction and recursion over higher-order encodings has been a long-standing challenge~\cite{despeyroux94higherorder,despeyroux:prirh-jv,DBLP:conf/tlca/DespeyrouxFH95,rockl:higoas}  with satisfactory solutions such as Beluga~\cite{pientka:typtfp,beluga} and Abella~\cite{BaeldeD:JFR} only appearing fairly recently.
Though this is not an insurmountable obstacle, supporting induction or recursion over HOAS has necessitated developing significant new logical or semantic foundations, for example \emph{contextual modal type theory} for Beluga~\cite{PfenningF:cmtt} and logics for reasoning about generic judgments and nominal abstraction for Abella~\cite{mcdowell02tocl,miller04tocl,tiu:logrgj,gacek:noma}.  Both of these go well beyond vanilla first-order logic so the prospect of combining them with matching logic seems much more speculative than for nominal logic.  Other techniques, such as locally nameless encodings~\cite{PollackR:metatheory}, appear to have no obvious obstacle to being adapted for use with matching logic, but there is no systematic logical foundation for locally nameless encodings analogous to nominal logic, instead one can view them as a recipe for employing de Bruijn indices for bound names while using ordinary names for free names such that induction and recursion principles can be proved.

The focus of this paper is on the question of how to adapt matching logic to accommodate binding.  Naturally, there are many techniques for encoding binding in theorem provers and logical frameworks that are not subject to this constraint, and we do not make any claim here that NML (or NLML) is better than the state of the art among such systems generally.  For this reason we have not attempted to give a complete picture of this field, but refer to recent comparisons such as~\cite{hoas-benchmarks} and \cite{poplmark-reloaded}.   As attempts to reason about languages with binding in matching logic or $\mathbb{K}$ are in their infancy, with the first proposal appearing only two years ago, it seems premature to try to predict whether the end results will be competitive with established approaches such as Abella~\cite{abella}, locally nameless in Coq~\cite{PollackR:metatheory}, Nominal Isabelle~\cite{UrbanC:ntih}, or Beluga~\cite{beluga}, which are already used by nonspecialists for formalizing their own research.  Each of these systems offers significant advantages that are not directly addressed by our work, for example support for automated reasoning, inductive proofs, and (for HOAS-based approaches) built-in capture-avoiding substitution.

Likewise, although implementation in concrete tools such as $\bbK$ is not in focus in this paper, for completeness we  mention some other comparable specification and verification tools in this space such as Ott~\cite{sewell10jfp}, PLT Redex~\cite{pltredex}, Bedwyr~\cite{bedwyr}, and \aprolog/$\alpha$Check~\cite{cheney:nomlp,cheney17tplp}.  Ott is a tool primarily aimed at easing the work involved in managing large semantics specifications, and provides a convenient concrete syntax for terms and rules that is parsed to abstract syntax that can then be transformed to LaTeX for presentation or publication purposes, or to specifications suitable for various theorem provers, such as Isabelle/HOL or Coq.  Ott provides some support for binding and can generate specifications that use popular techniques such as locally-nameless~\cite{PollackR:metatheory}, but focuses on specifying and representing syntax and rules, not constructing or checking proofs, and Ott does not have an underlying logic in which properties of Ott specifications could be specified or proved.  Likewise, PLT Redex~\cite{pltredex} is a domain-specific language built on Racket that supports specifying languages with conventional abstract syntax, evaluation rules and typing rules.  It is again primarily focused on making it easy to define and experiment with specifications, including visualization of how expressions evaluate, but its support for binding is somewhat limited: binding operators are recognized to some extent and capture-avoiding substitution is provided as a built-in operation.
Bedwyr~\cite{bedwyr} is a system based on the logic of generic judgments~\cite{tiu:logrgj} for specifying and model-checking properties of formal calculi including binding.  Finally, \aprolog is a typed logic programming language that provides built-in support for name binding using nominal logic.  \aprolog is not a full-fledged logical framework but can be used to model and check properties of specifications, using an exhaustive form of property-based testing~\cite{cheney17tplp}.  Adapting ideas from (nominal) matching logic to logic programming may be an interesting future direction.

The variants of matching logic defined in this paper include names and name management primitives following the nominal approach, which supports inductive reasoning over syntax with binding. We propose three alternative ways of incorporating nominal primitives in matching logic:  in the first one (section~\ref{sec:nomlm}), we exploit the fact that nominal logic has been defined as a theory of first-order logic, which is definable in matching logic, in the second (section~\ref{sec:sndapproach}), we use the features of patterns to define nominal primitives, and in the third (section~\ref{sec:newnml}), we extend the syntax and the semantics of matching logic to include built-in nominal primitives. Each of these approaches permits to encode binders directly by using the abstraction construct of nominal logic, avoiding the use of functional encodings and inheriting  the reasoning principles of nominal logic. 

The treatment of concretion and locally scoped fresh names in NML is somewhat similar to the $\lambda\alpha\nu$-calculus of \citet{pitts11jfp}.  The latter is a lambda-calculus extended with a binding abstraction operator $\alpha a.e$, concretion $e@a$, and locally scoped names $\nu a. e$; expressions are given a total semantics by interpreting over nominal sets enriched with an operation called \emph{restriction}, which models the behavior of $\nu$.  Thus, in contrast to NML, concretion is total and satisfies a beta-equivalence law $(\alpha a. e)@b = \nu a. \swap{a}{b}{e}$, but free occurrences of $a$ in the left hand side (obtained by swapping with occurrences of $b$ in $e$) will be replaced by a special name $\mathsf{anon} = \nu a. a$.  Our treatment of concretion and local fresh name scoping using $\new$ leverages matching logic's ability to represent partial functions.

\section{Conclusions}\label{sec:concl}

We have illustrated the expressive power of NML by using it to specify the $\beta$-reduction and typing relations in the $\lambda$-calculus as well as  principles of induction to reason over $\lambda$-terms. In future we will study extensions of NML with fixed point operators, and will consider proof systems and computational properties (e.g., algorithms to  check equivalence of patterns in NML). This is important to build theorem provers and programming language verification tools based on this logic. 
Regarding implementations of NLML and NML: the first can be directly implemented in $\mathbb{K}$ since it is simply a theory in matching logic; the second requires an extension of  $\mathbb{K}$, which we are discussing  with the designers of $\mathbb{K}$. 
We expect the efficient techniques developed for nominal terms with ground names can be adapted to work with NML patterns.

\section*{Acknowledgments}
Thanks to Xiaohong  Chen for helpful comments.  This work was supported by ERC Consolidator Grant Skye (grant number 682315) and by a Royal Society International Exchange grant (number IES$\backslash$R2$\backslash$212106).

\bibliography{main}
\bibstyle{plainurl}
\bibliographystyle{ACM-Reference-Format}
\newpage
\appendix
\onecolumn

\section{Proofs}\label{sec:proofs}

\begin{proof}[Proof of Prop.~\ref{prop:axiom-translation}]
It is easy to see that axioms $(E3)$ and $(E5)$ are particular cases of axiom $(EV)$: function symbols and predicates are symbols in the signature $\Sigma_{NLML}$.

Axiom $(E1)$ is also a particular case of axiom $(EV)$ since swapping is represented as a symbol in the signature $\Sigma_{NLML}$ (it is a particular $\sigma$ of arity 3: two names and a pattern in which we are swapping the names).

Axiom $(E2)$ does not exactly translate to  $(EV)$, but using axiom $(P)$ we can show that $a \fresh x = \swap{b}{b'}{(a \fresh x)}$ and then using $(EV)$ the desired result follows.  Axiom $(E4)$ is handled similarly.
\end{proof}

\begin{proof}[Proof of Prop.~\ref{prop:translation}]
The  translation function is defined inductively in the obvious way using the mapping between symbols in the nominal logic and NLML signatures described above, resulting in a pattern of sort $Pred$ since this is the type of the result for predicate symbols and formula constructors.  

Part 1 follows from the fact that FOL can be represented as a theory in matching logic using the $Pred$ sort introduced above  (see~\citet{chen19lics} for the definition of the FOL theory in matching logic), and the axioms of nominal logic can be represented as a theory in matching logic too, as a consequence of Proposition~\ref{prop:axiom-translation}.

Part 2 follows from completeness in first-order logic and nominal logic~\cite{cheney:comhtn} using $\neg \phi$: the formula $\neg \phi$ is provable, and by part 1, $\neg \phi' = \top$ holds in matching logic, and therefore $\phi' = \bot$ holds also.
\end{proof}

\begin{proof}[Proof of Prop.~\ref{prop:abstraction}]
Using $(*)$ we can write $(A1')$ as follows
\[(\exists y. y \wedge \abs{b}{y} = \abs{a}{x}) =  (\swap{a}{b}{x} \wedge b \coerce{\fresh} \abs{a}{x})\]
Note that $b \fresh \abs{a}{x} = (b  = a \vee b \fresh x)$, so the above pattern is equal to:
\[((\exists y. y \wedge \abs{b}{y} = \abs{a}{x}) = (\swap{a}{b}{x} \wedge a \coerce{=} b)) \vee 
((\exists y. y \wedge \abs{b}{y} = \abs{a}{x}) = (\swap{a}{b}{x}\wedge b \coerce{\fresh} x))\]
By $(A1)$ in the first disjunct we must have $x = y$ so the formula simplifies to $\top$ using $(S1)$.  In the second disjunct, using $(A1)$ we must have $y = \swap{a}{b}{x} \wedge b \coerce{\fresh} x$, so the formula again simplifies to $\top$ by instantiating $y$.

Using $(*)$ we can write $(A2')$ as follows:
\[x = \exists a. a \coerce{\fresh} x \wedge \abs{a}(\exists y. y \wedge \abs{a}{y} = x)
\]
Using $(A2)$ we know we can choose $b,z$ such that $x = \abs{b}{z}$.  We thus need to show that for any such choice of $b,z$ we have:
\[\abs{b}{z} = \exists a. a \coerce{\fresh} \abs{b}{z} \wedge \abs{a}{(\exists y. y \wedge \abs{a}{y} = 
\abs{b}{z})}
\]
We can again consider two cases: $a = b$ or $a \fresh z$.  In the first case the formula simplifies as follows:
\[
\abs{b}{z} = \abs{b}(\exists y. y \wedge \abs{b}{y} = 
\abs{b}{z})
\]
and by $(A1)$ we have $y = z$ so the whole formula simplifies to $\top$ on instantiating $y = z$.  In the second case the formula simplifies to:
\[\abs{b}{z} = \exists a. a \coerce{\fresh} z \wedge \abs{a}{(\exists y. y \wedge \abs{a}{y} = 
\abs{b}{z})}
\]
where using $(A1)$ we can replace $\abs{a}{y} = \abs{b}{z}$ with $y = \swap{a}{b}{z}\wedge a \fresh z$.  Thus the formula simplifies further to
\[\abs{b}{z} = \exists a. a \coerce{\fresh} z \wedge \abs{a}{\swap{a}{b}{z}} )\]
To complete the proof we note that an $a$ fresh for $z$ can always be chosen, discharging the freshness assumption and again using $(A1)$ we can conclude that $\abs{b}{z} = \abs{a}{\swap{a}{b}{z}}$ since $b \fresh \swap{a}{b}{z}$ follows from $a \fresh z$ by equivariance.

To show $(A1)$ we need to show:
\[\abs{a}{x} = \abs{b}{y} \iff (a = b \wedge x = y) \vee (a \fresh y \wedge x = \swap{a}{b}{y})\]
In the forward direction, suppose $a,b,x,y$ given with $\abs{a}{x}= \abs{b}{y}$.  Let $c\fresh a,b,x,y$ be a fresh name; then $(\abs{a}{x})@c = (\abs{b}{y})@c$.  Using $(A1')$ this equation is equivalent to $(\swap{a}{c}{x} \wedge c \fresh \abs{a}{x}) = (\swap{b}{c}{y} \wedge c \fresh \abs{b}{y})$.  Since $c$ is fresh for everything in sight the freshness constraints simplify away and we have $\swap{a}{c}{x} = \swap{b}{c}{y}$.  We now consider two cases.  If $a = b$ then $\swap{a}{c}{x}= \swap{a}{c}{y}$ implies $x = y$.  Otherwise, notice that $a \fresh \abs{a}{x} = \abs{b}{y}$ and if $a \neq b$ the only way this can be the case is if $a \fresh y$.  This, together with properties of swapping, means we can conclude
$\abs{c}{(\swap{a}{c}{x})} = \abs{c}{\swap{a}{c}{\swap{a}{b}{y}}} =\abs{c}{\swap{c}{b}{\swap{a}{c}{y}}} = \abs{\swap{c}{b}{b}}{\swap{c}{b}{y}} = \swap{c}{b}{\abs{b}{y}}$, and since $a \fresh y$ we  know $c \fresh \abs{b}{y}$ so $\swap{c}{b}{\abs{b}{y}} = \abs{b}{y}$.
Otherwise we may assume $c \neq a$ and $c \fresh x$ and already know $b \neq a$ so both $b$ and $c$ are fresh for $\abs{c}{\swap{a}{c}{x}}$ so we have $\abs{c}{(\swap{a}{c}{x})} = \swap{b}{c}\abs{c}{\swap{a}{c}{x}} = \abs{b}{\swap{b}{c}{\swap{a}{c}{x}}} = \abs{b}{\swap{a}{b}{x}} = \abs{b}{\swap{a}{b}{\swap{a}{b}{y}}} = \abs{b}{y}$.  Since $c$ was arbitrary, we can conclude that $\exists c. c \fresh \abs{a}{x} \wedge \abs{c}{\swap{a}{c}{x}} \subseteq \abs{b}{y}$.  The converse direction, showing that $\exists c. c \fresh \abs{b}{y} \wedge \abs{c}{\swap{b}{c}{y}} \subseteq \abs{a}{x}$, is symmetric noting that $a \fresh x \wedge x = \swap{a}{b}{y}$ is equivalent to $b \fresh x\wedge y = \swap{a}{b}{x}$.

To show $(A2)$ suppose $x$ is a given abstraction and use $(A2')$ to choose a fresh name $a \fresh x$ such that $x = \abs{a}{x@a}$; then $(A2)$ holds with the existential quantifiers witnessed by $a$ and $y = x@a$.
\end{proof}

\begin{proof}[Proof of Prop.~\ref{prop:fresh}]Each primed axiom is equivalent to the corresponding nominal logic axiom.  For example $(F1')$ just says that swapping any two names matching $fresh(x)$ (i.e., both fresh for $x$) does not affect $x$.  $(F2')$ just says that the names of sort $\nsort$ fresh for $a:\nsort$ are all those not matching $a$, while $(F3')$ says that the names of a different sort $\nsort'$ fresh for $a:\nsort$ are all names.  The final axiom $(F4')$ just says $fresh(x)$ is nonempty for any $x$, i.e., a fresh name can always be chosen.  Note that since we assume product sorts are part of any signature, $\sort$ may be a product of finitely many sorts, enabling us to emulate instances of axiom scheme $(F4)$ with multiple variables.
\end{proof}

\begin{proof}[Proof of Prop.~\ref{prop:abs-wedge-not}]
For part (1), suppose $\rho$ is a suitable valuation.  Then
\begin{eqnarray*}
\semantics{\rho}{[a](\phi_1 \wedge \phi_2)} 
&=& \{\abs{a'}{v} \mid v \in \semantics{\rho}{\phi_1 \wedge \phi_2}\}\\
&=& \{\abs{a'}{v} \mid  v \in \semantics{\rho}{\phi_1} \cap \semantics{\rho}{\phi_2}\}\\
&=& \{\abs{a'}{v} \mid v \in \semantics{\rho}{\phi_1}\} \cap 
\{\abs{a'}{v} \mid v \in \semantics{\rho}{\phi_2}\}\\
&=& \semantics{\rho}{\abs{a}\phi_1}\cap \semantics{\rho}{\abs{a}{\phi_2}}\\
&=& \semantics{\rho}{\abs{a}\phi_1 \wedge \abs{a}{\phi_2}}
\end{eqnarray*}
and since $\rho$ is arbitrary we can conclude $[a](\phi_1 \wedge \phi_2) \iff [a]\phi_1 \wedge [a]\phi_2$.

For part (2), similarly suppose suitable $\rho$ is given and suppose $a:\nsort$ and $\phi : \sort$. Let $a' = \rho(a)$.  Then 
\begin{eqnarray*}
\semantics{\rho}{[a](\neg \phi)} 
&=& \{\abs{a'}{v} \mid  v \in \semantics{\rho}{\neg \phi}\}\\
&=& \{\abs{a'}{v} \mid  v \in M_{\sort}-\semantics{\rho}{\phi}\}\\
&=& \{\abs{a'}{v} \mid v \in M_{\sort}\}- \{\abs{a'}{v} \mid  v \in \semantics{\rho}{\phi}\}\\
&=& \semantics{\rho}{\abs{a}{\top}}- \semantics{\rho}{\abs{a}{\phi}}\\
&=& \semantics{\rho}{\abs{a}{\top} \wedge \neg[a]\phi}
\end{eqnarray*}
Since $\rho$ was arbitrary we can conclude  $[a](\neg \phi) \iff \abs{a}{\top} \wedge \neg[a]\phi$.

Finally for part (3) consider again a suitable $\rho$ and let $a' = rho(\Aa)$ and $b' = \rho(\Ab)$.  Recall that since $\Aa$ and $\Ab$ are  distinct names we must also have $a' \neq b'$ since $\rho$ must be injective on names.
\begin{eqnarray*}
\semantics{\rho}{(\abs{\Aa}{\phi}) @ \Ab}
&=& \{v \mid \abs{b'}{v} \in \semantics{\rho}{\abs{\Aa}{\phi}}\\
&=& \{v \mid \abs{b'}{v} \in \{\abs{a'}{w} \mid w \in \semantics{\rho}{\phi}\}\\
&=& \{v \mid w \in \semantics{\rho}{\phi}, \abs{b'}{v} = \abs{a'}{w} \}\\
&=& \{v \mid w \in \semantics{\rho}{\phi}, b' \fresh w \wedge v = \swap{a'}{b'}{w} \}\\
&=&  \{v \in M_{\abs{\nu}{\sort}} \mid w \in \semantics{\rho}{\phi}, b' \fresh \abs{a'}{w}\} \cap  \{\swap{a'}{b'}{v} \mid v \in \semantics{\rho}{\phi}\\
&=&  \{v \in M_{\abs{\nu}{\sort}} \mid \star \in \{\star \mid w \in \semantics{\rho}{\phi} \mid b' \fresh \abs{a'}{w}\}\}\cap \semantics{\rho}{\swap{\Aa}{\Ab}{\phi}}\\
&=&  \{v \in M_{\abs{\nu}{\sort}} \mid \star \in \semantics{\rho}{\Ab \fresh \abs{\Aa}{\phi}}\} \cap \semantics{\rho}{\swap{\Aa}{\Ab}{\phi}}\\
&=&  \semantics{\rho}{\Ab \fresh \abs{\Aa}{\phi}} \cap \semantics{\rho}{\swap{\Aa}{\Ab}{\phi}}\\
&=&  \semantics{\rho}{(\Ab \fresh \abs{\Aa}{\phi}) \wedge \swap{\Aa}{\Ab}{\phi}}
\end{eqnarray*}
Since $\rho$ was arbitrary, this completes the proof that $(\abs{\Aa}{\phi}) @ \Ab \iff (\Ab \fresh \abs{\Aa}{\phi}) \wedge \swap{\Aa}{\Ab}{\phi}$.
\end{proof}

\begin{proof}[Proof of Thm.~\ref{lem:equivariant-sem}]
By induction on the definition of patterns. 
\begin{itemize}
\item 
The cases for names and variables are trivial: by Definition~\ref{def:NML-pattern-sem}, if $e$ is a name or a variable then $\swap{a}{a'}{\semantics{\rho}{e}} = \swap{a}{a'}{\{\rho(e)\}} = \{\swap{a}{a'}{\rho(e)}\} =  \semantics{\swap{a}{a'}{\rho}}{e}$.
\item
If $\phi_\sort = \sigma(\phi_{\sort_1},\ldots,\phi_{\sort_n})$ the result follows by induction and the fact that the interpretation of $\sigma$ is equivariant (see Def.~\ref{def:NML-model}). 
\item
If $\phi_\sort = \phi_1 \wedge \phi_2$ or $\phi_\sort = \neg \phi'$ the result follows directly by induction.  
For example 
\begin{eqnarray*}
\swap{a}{a'}{\semantics{\rho}{\phi \wedge \phi'}} 
&=& \swap{a}{a'}{(\semantics{\rho}{\phi} \cap \semantics{\rho}{\phi'})} \\
&=&\swap{a}{a'}{\semantics{\rho}{\phi} }\cap  \swap{a}{a'}{\semantics{\rho}{\phi'}} \\
&=&  \semantics{\swap{a}{a'}{\rho}}{\phi}\cap  \semantics{\swap{a}{a'}{\rho}}{\phi'} \\
&=& \semantics{\swap{a}{a'}{\rho}}{\phi \wedge \phi'}\;.
\end{eqnarray*}
\item For existential patterns $\exists x:\sort. \phi$ we reason as follows.
\[\begin{array}{lcl}
\swap{a}{a'}{\semantics{\rho}{\exists x:\sort.\phi}}
&=&  \swap{a}{a'}{\bigcup_{v \in M_\sort}\semantics{\rho[v/x]}{\phi}}
=\bigcup_{v \in M_\sort}\swap{a}{a'}{\semantics{\rho[v/x]} {\phi}}\\
&=&  \bigcup_{v \in M_\sort}\semantics{\swap{a}{a'}{\rho[v/x]}}{\phi}\\
&=&  \bigcup_{\swap{a}{a}{v} \in \swap{a}{a'}{M_\sort}}\semantics{(\swap{a}{a'}{\rho})[\swap{a}{a'}{v}/x]}{\phi}\\
 &= & \bigcup_{w \in M_\sort}\semantics{(\swap{a}{a'}{\rho})[w/x]}{\phi}\\
&=& \semantics{\swap{a}{a'}{\rho}}{\exists x:\sort.\phi}
\end{array}\]
where in addition to the induction hypothesis we use the fact that a nominal set $M_\sort$ is closed under swappings so that $M_\sort = \swap{a}{b}{M_\sort}$.

\item For the case of $\new \Aa:\nsort.\phi$ we reason as follows.
\[\begin{array}{lcl}
\swap{a}{a'}{\semantics{\rho}{\new \Aa:\nsort.\phi}}
&=& \swap{a}{a'}{\bigcup_{b \in \bbA_{\nsort} - \supp(\rho)} \{v \in \semantics{\rho[b/\Aa]}{\phi}} \mid b \not\in supp(v)\}\\
&= & \bigcup_{b \in \bbA_{\nsort} - \supp(\rho)} \{\swap{a}{a'}{v} \in \swap{a}{a'}{\semantics{\rho[b/\Aa]}{\phi}} \mid b \not\in supp(v)\}\\
&=& \bigcup_{b \in \bbA_{\nsort} - \supp(\rho)} \{ \swap{a}{a'}{v} \in \semantics{\swap{a}{a'}{\rho[b/\Aa]}}{\phi} \mid b \not\in supp(v)\}\\
&=& \bigcup_{\swap{a}{a'}{b} \in \swap{a}{a'}{(\bbA_{\nsort} - \supp(\rho))}}\{ \swap{a}{a'}{v} \in \semantics{(\swap{a}{a'}{\rho})[\swap{a}{a'}{b}/\Aa]}{\phi} \mid  b \not\in supp(v)\} \\
&=& \bigcup_{b' \in \bbA_{\nsort}-\supp(\swap{a}{a'}{\rho})} \{w \in \semantics{(\swap{a}{a'}{\rho})[b'/\Aa]}{\phi} \mid b'\not\in supp(w)\}\\
&=& \semantics{\swap{a}{a'}{\rho}}{\new \Aa:\nsort.\phi}
\end{array}\]
Here we make use of the fact that $M_{\nsort} = \bbA_{\nsort}$ is a nominal set hence closed under swappings, as well as the fact that set difference and the support operation are equivariant.

\end{itemize}
\end{proof}

\begin{proof}[Proof of Prop.~\ref{prop:new-reordered}]
The first part is a direct consequence of the definition of $\semantics{\rho}{\new \Aa.\phi}$ (Definition~\ref{def:NML-pattern-sem}): if $\Aa$ is not free in $\phi$ then $\semantics{\rho}{\phi} = \semantics{\rho[a/\Aa]}{\phi}$ for any $a$.  

The second part also follows from Definition~\ref{def:NML-pattern-sem}: Both patterns are interpreted by
$$\bigcup_{\small\begin{array}{c}a \in \mathbb{A}_{\nsort}-supp(\rho)\\ b \in  \mathbb{A}_{\nsort'}-supp(\rho)\\ a\neq b\end{array}}\{v\in \semantics{\rho[a/\Aa][b/\Ab]}{\phi_{\sort}}\mid a \not\in supp(v),  b \not\in supp(v)\}$$
since the conditions on $a, b$ ensure $\swap{a}{b} v = v$ and $\swap{a}{b}{\rho[a/\Aa][b/\Ab]} = \rho[a/\Ab][b/\Aa]$.
\end{proof}

\begin{proof}[Proof of Prop.~\ref{prop:new-wedge-not}]
\begin{enumerate}
    \item 
We can show that both patterns have the same semantics using Theorem~\ref{lem:equivariant-sem} (Equivariance).

By Definition~\ref{def:NML-pattern-sem} (semantics of $\new$ and $\wedge$ patterns):
$$\semantics{\rho}{\new \Aa\colon \nsort. \phi_{\sort} \wedge \psi_\sort} = 
\bigcup_{a \in \mathbb{A}_{\nsort}-supp(\rho)}\{v\in \semantics{\rho[a/\Aa]}{\phi_{\sort}}\cap \semantics{\rho[a/\Aa]}{\psi_{\sort}} \mid a \not\in supp(v)\}.$$
The latter is equal to  $\bigcap_{a \in \mathbb{A}_{\nsort}-supp(\rho)}\{v\in \semantics{\rho[a/\Aa]}{\phi_{\sort}}\cap \semantics{\rho[a/\Aa]}{\psi_{\sort}} \mid a \not\in supp(v)\}$. To see that the union can be turned into an intersection,  it is sufficient to consider two different $a,a'$ in $\mathbb{A}_{\nsort}-supp(\rho)$ such that $a, a' \not\in supp(v)$ and observe that 
$v \in \semantics{\rho[a/\Aa]}{\phi_{\sort} }$
if and only if $v = \swap{a}{a'}{v} \in \semantics{\swap{a}{a'}{(\rho[a/\Aa])}}{\phi_{\sort}}= \semantics{\rho[a'/\Aa]}{\phi_{\sort}}$ by Theorem~\ref{lem:equivariant-sem} (and similarly for $\psi_\sort$).
Therefore
\begin{eqnarray*}
\semantics{\rho}{\new \Aa\colon \nsort. \phi_{\sort} \wedge \psi_\sort} 
&=& \bigcap_{a \in \mathbb{A}_{\nsort}-supp(\rho) }\{v\in \semantics{\rho[a/\Aa]}{\phi_{\sort}}\mid a \not\in supp(v)\} \cap\\ 
& & \bigcap_{a \in \mathbb{A}_{\nsort}-supp(\rho) }\{v\in  \semantics{\rho[a/\Aa]}{\psi_{\sort}} \mid a \not\in supp(v)\}.
\end{eqnarray*}
Using the same reasoning again, we can write: 
\begin{eqnarray*}
\semantics{\rho}{\new \Aa\colon \nsort. \phi_{\sort} \wedge \psi_\sort} 
&=& \bigcup_{a \in \mathbb{A}_{\nsort}-supp(\rho) }\{v\in \semantics{\rho[a/\Aa]}{\phi_{\sort}} \mid a \not\in supp(v)\} \cap\\ 
& & \bigcup_{a \in \mathbb{A}_{\nsort}-supp(\rho) }\{v\in  \semantics{\rho[a/\Aa]}{\psi_{\sort}} \mid a \not\in supp(v)\}\\
&=& \semantics{\new \Aa\colon \nsort.\phi_\sort \wedge \new \Aa\colon \nsort.\psi_\sort}{\rho}\;.
\end{eqnarray*}

Notice that by Prop.~\ref{prop:new-reordered}, $\new \Aa. (\phi ~\wedge ~\psi) \iff ((\new \Aa.\phi)~ \wedge ~\psi)$ if $\Aa$ is not free in $\psi$.

\item This case also relies on the equivariance of the semantics (Theorem~\ref{lem:equivariant-sem}):
\begin{eqnarray*}
\semantics{\rho}{\neg \new \Aa\colon \nsort. \phi_{\sort}}&= & M_\sort - \semantics{\rho}{\new \Aa\colon \nsort. \phi_{\sort}}\\
&=&
M_\sort - \bigcup_{a \in \mathbb{A}_{\nsort}-supp(\rho) } \{v\in \semantics{\rho[a/\Aa]}{\phi_{\sort}} \mid a \not\in supp(v)\}
\end{eqnarray*}

\begin{eqnarray*}
\semantics{\rho}{\new \Aa\colon \nsort. \neg\phi_{\sort}}&=&
\bigcup_{a \in \mathbb{A}_{\nsort}-supp(\rho) } \{v\in \semantics{\rho[a/\Aa]}{\neg\phi_{\sort}} \mid a \not\in supp(v)\}\\&=&
\bigcup_{a \in \mathbb{A}_{\nsort}-supp(\rho) } \{v\in M_\sort - \semantics{\rho[a/\Aa]}{\phi_{\sort}} \mid a \not\in supp(v)\}\\&=& 
M_\sort - \bigcap_{a \in \mathbb{A}_{\nsort}-supp(\rho) } \{v\in \semantics{\rho[a/\Aa]}{\phi_{\sort}} \mid a \not\in supp(v)\}
\end{eqnarray*}
The union and intersection above are equivalent, as shown in the previous case:
If we consider two different elements $a$ and $a'$ in $\mathbb{A}_{\nsort}-supp(\rho)$ such that $a, a' \not\in supp(v)$ then $v \in \semantics{\rho[a/\Aa]}{\phi_{\sort} }$
if and only if $v = \swap{a}{a'}{v} \in \semantics{\swap{a}{a'}{(\rho[a/\Aa])}}{\phi_{\sort}}= \semantics{\rho[a'/\Aa]}{\phi_{\sort}}$ (Theorem~\ref{lem:equivariant-sem}).

\item
By Definition~\ref{def:NML-pattern-sem},
\begin{eqnarray*}
\semantics{\rho}{\abs{\Aa}{\new \Ab\colon \nsort.\phi_\sort}} &=& \abs{\rho(\Aa)}{\bigcup_{b\in \mathbb{A}_{\nsort}-supp(\rho)} \{v \in \semantics{\rho[b/\Ab]}{\phi_\sort} \mid b \not\in supp(v)\}} \\
&=& \bigcup_{b\in \mathbb{A}_{\nsort}-supp(\rho)} \{\abs{\rho(\Aa)}{v} \mid v \in \semantics{\rho[b/\Ab]}{\phi_\sort}, b \not\in supp(v)\}
\end{eqnarray*}

Also by Definition~\ref{def:NML-pattern-sem},
\begin{eqnarray*}
\semantics{\rho}{\new \Ab\colon \nsort.\abs{\Aa}{\phi_\sort}} &=&
\bigcup_{b\in \mathbb{A}_{\nsort}-supp(\rho)}\{w\in \semantics{\rho[b/\Ab]}{\abs{\Aa}{\phi_\sort}} \mid b \not \in supp(w)\} \\
& = &
\bigcup_{b\in \mathbb{A}_{\nsort}-supp(\rho)}\{ \abs{\rho(\Aa)}{v} \mid v \in \semantics{\rho[b/\Ab]}{\phi_\sort}, b \not \in supp(v)\}
\end{eqnarray*}

\item  
By Definition~\ref{def:NML-pattern-sem},
$\semantics{\rho}{\langle \new \Aa. \phi, \psi \rangle} = \langle \bigcup_{a \in \mathbb{A}_{\nsort} - supp(\rho)} \{v \in \semantics{\rho[a/\Aa]}{\phi} \mid a \not\in supp(v) \}, \semantics{\rho}{\psi} \rangle$.

Since $\Aa$ does not occur free in $\psi$, $\semantics{\rho[a/\Aa]}{\psi}=  \semantics{\rho}{\psi}$ for any $a$.
Hence, 
\[\semantics{\rho}{\langle \new \Aa. \phi, \psi \rangle} = \bigcup_{a \in \mathbb{A}_{\nsort} - supp(\rho)} 
\{ \langle  v, w \rangle \mid v \in \semantics{\rho[a/\Aa]}{\phi},  w\in \semantics{\rho[a/\Aa]}{\psi}, a \not\in supp(v) \}\;.\]
Also by Definition~\ref{def:NML-pattern-sem}, 
\[ \semantics{\rho}{\new \Aa. \langle \phi,\psi\rangle} =   \bigcup_{a \in \mathbb{A}_{\nsort} - supp(\rho)} 
\{ \langle  v, w \rangle \mid v \in \semantics{\rho[a/\Aa]}{\phi},  w\in \semantics{\rho[a/\Aa]}{\psi}, a \not\in supp(v,w) \}\;.\]

The only difference between the two sets is the fact that in the second one we require $a \not\in supp(w)$. Therefore every element in the second set is also in the first. We need to show that every element in the first set is also in the second: 
Assume $\langle  v, w \rangle$ is in the first set and was obtained using $a$. If $a\not\in supp(v,w)$ then $\langle  v, w \rangle$ is also in the second set and we are done. If  $a \in supp(w)$ there exists $a'$ such that $a' \not\in supp(w,v,\rho)$.
Since $\Aa$ is not free in $\psi$, $\semantics{\rho[a'/\Aa]}{\psi} = \semantics{\rho[a/\Aa]}{\psi}$, hence $w \in \semantics{\rho[a'/\Aa]}{\psi}$. By equivariance, $\swap{a}{a'}{v} \in  \semantics{\swap{a}{a'}{\rho[a/\Aa]}}{\phi} = \semantics{\rho[a'/\Aa]}{\phi}$. Hence, $\langle \swap{a}{a'}{v}, w\rangle$ is in the second set and this is $\langle v, w\rangle$ since $a,a'\not\in supp(v)$. 

In the more general case:

By Definition ~\ref{def:NML-pattern-sem},
\begin{eqnarray*}
\semantics{\rho}{\sigma(\new \Aa. \phi_1,\ldots, \new \Aa. \phi_n)} & = &\sigma_M(\semantics{\rho}{\new \Aa.\phi_1}, \ldots, \semantics{\rho}{\new \Aa.\phi_n}) \\
&=&
 \sigma_M\left( \bigcup_{a \in \mathbb{A}_{\nsort} -  supp(\rho)} \{v \in \semantics{\rho[a/\Aa]}{\phi_1} \mid a \notin supp(v)\}, \ldots,\bigcup_{a \in \mathbb{A}_{\nsort} - supp(\rho)} \{v \in\semantics{\rho[a/\Aa]}{\phi_n})\mid a \not\in supp(v)\}\right)
\end{eqnarray*}
and
\[\semantics{\rho}{\new \Aa.\sigma(\phi_1,\ldots,\phi_n)} = \bigcup_{a \in \mathbb{A}_{\nsort} - supp(\rho)} \{ v \in \sigma_M(\semantics{\rho[a/\Aa]}{\phi_1} , \ldots, \semantics{\rho[a/\Aa]}{\phi_n})\mid a \not\in supp(v)\}\;.\]

Again we need to show that both sets are the same, and in one direction the result is direct (if an atom is not in the support of $v_1,\ldots, v_n$, it cannot be in the support of $\sigma_M(v_1,\ldots,v_n)$ since $\sigma_M$ is equivariant). In the other direction, the result follows from the assumption 
 $\forall x_1,\ldots,x_n, \Aa, \Aa \# \sigma(x_1,...,x_n) \Rightarrow  \Aa \# x_1 \wedge \ldots \wedge \Aa \# x_n$: if $a$ is not  in the support of $\sigma_M(v_1,\ldots,v_n)$ then it is not in the support of $v_1, \ldots, v_n$.

\end{enumerate}
\end{proof}

\begin{proof}[Proof of Prop.~\ref{prop:new-derived}]
To show that two patterns $\phi_\sort$ and $\psi_\sort$ are  equivalent, we need to show they have the  same instances in $M_\sort$ for any given $\rho$.

To prove the first equivalence,  assume $v\in M_\sort$ matches $\new \Aa\colon \nsort. \phi_\sort$, that is, there is  some $a\in \mathbb{A}_{\nsort} - supp(\rho)$ such that
$v\in \semantics{\rho[a/\Aa]}{\phi_{\sort}}$ and $a \not\in supp(v)$. Hence, there exists $a\in \mathbb{A}_{\nsort}$ such that $a$ is fresh for $\rho(\vec{\Ab},\vec{x}),v$ and $v\in \semantics{\rho[a/\Aa]}{\phi_{\sort}}$. Therefore, by definition of $\rho$ (existential and conjunction cases), we deduce that $v$ matches $\exists z_\Aa\colon \nsort. ((\exists y\colon \sort. y \wedge z_\Aa \# (\vec{\Ab},\vec{x},y)) \wedge \phi_\sort\{\Aa \mapsto z_{\Aa}\}$.

Similarly, if $v\in M_\sort$ matches $\exists z_\Aa\colon \nsort. ((\exists y\colon \sort. y \wedge z_\Aa \# (\vec{\Ab},\vec{x},y)) \wedge \phi_\sort\{\Aa \mapsto z_{\Aa}\}$, then there is a value $a \in \mathbb{A}_{\nsort}$ such that $v$ matches both conjuncts, that is, $a$ is fresh for $\semantics{\rho}{\vec{\Ab},\vec{x}}$, $v$  and $v$ is in $\semantics{\rho[a/z_\Aa]}{\phi_\sort\{\Aa \mapsto z_{\Aa}\}}$. If $a$ is fresh for $supp(\rho)$ then we are done, as  
$v\in \semantics{\rho}{\new \Aa\colon \nsort. \phi_\sort}$. If $a$ is not fresh for $supp(\rho)$ (i.e., it occurs in the image of some variable or name that does not occur free in $\phi_\sort$), there exists $a'$ such that $a'\in \mathbb{A} -\supp(\rho)$  and  $a'$ is also fresh for $v$ (hence $\swap{a}{a'}{v} = v$). By the equivariance theorem, $\semantics{\swap{a}{a'}{(\rho[a/z_{\Aa}])}}{\phi_\sort\{\Aa \mapsto z_{\Aa}\}} = \swap{a}{a'}{v} = v$. Since  $\semantics{\swap{a}{a'}{(\rho[a/z_{\Aa}])}}{\phi_\sort\{\Aa \mapsto z_{\Aa}\}}= \semantics{\rho[a'/z_{\Aa}]}{\phi_\sort}$ (because $a' \not\in supp(\rho)$ and $a$ is not in the support of free names or variables in $\phi_\sort$), we deduce $v\in \rho(\new \Aa\colon \nsort.\phi_\sort)$ as required.

In the second equivalence,  $\Rightarrow$ is the pattern implication operation, defined as $\phi \Rightarrow \psi \Leftrightarrow \neg(\phi) \vee \psi$. The semantics of $\phi \Rightarrow \psi$ is $M_\sort - (\semantics{\rho}{\phi} - \semantics{\rho}{\psi})$, i.e.,  the set of values that either don't match $\phi$ or do match $\psi$, i.e. the set of values $v$ such that if $v$ matches $\phi$ then $v$ matches $\psi$ also.  
The  pattern $\forall x:\sort'.\phi_\sort$ is syntactic sugar for $\neg \exists x: \sort'. \neg \phi_\sort(x)$, interpreted as the complement of the union for each $v \in M_{s'}$ of values in $M_\sort$ that do not match  $\semantics{\rho[v/x]}{\phi}$.

Assume $v\in M_\sort$ matches $\new \Aa\colon \nsort. \phi_\sort(\Aa,\vec{\Ab},\vec{x})$. Then there is a value $a \in \mathbb{A}_{\nsort}$ such that  $a$ is fresh for $\semantics{\rho}{\vec{x}}$,  $\semantics{\rho}{\vec{\Ab}}$, $v$  and $v$ is in $\semantics{\rho[a/z_\Aa]}(\phi_\sort(z_\Aa,\vec{\Ab},\vec{x}))$.

Consider now a different value $a' \in \mathbb{A}_{\nsort}$, we want to prove that $v$ still satisfies the implication pattern. 
If $v$ does not satisfy the existential for $a'$ (i.e., $a'$ is not fresh  for $\rho(\vec{x}),\rho(\vec{\Ab}), v$) then this is trivially true. Otherwise,  $a'$ is fresh for $\rho(\vec{x}),\rho(\vec{\Ab}), v$ and by equivariance (Theorem~\ref{lem:equivariant-sem}) $v$  must also match $\semantics{\rho[a'/z_\Aa]}(\phi_\sort (z_\Aa,\vec{\Ab},\vec{x}))$.

In the other direction, assume $v\in M_\sort$ matches 
$\forall z_\Aa\colon \nsort. ((\exists y\colon \sort. y \wedge z_\Aa \# (\vec{x},\vec{\Ab},y)) \Rightarrow \phi_\sort(z_\Aa,\vec{\Ab},\vec{x}))$.
The nominal set semantics ensures that there exists some $a\in\mathbb{A}_{\nsort}$ that is fresh for $\rho(\vec{x}),\rho(\vec{\Ab}),v$, that is, $v$ matches the left-hand side of the implication pattern, therefore $v$ must match  $\semantics{\rho[a/z_\Aa]}(\phi_\sort (z_\Aa,\vec{\Ab},\vec{x}))$. Hence, $v$ matches $\exists z_\Aa\colon \nsort. ((\exists y\colon \sort. y \wedge z_\Aa \# (\vec{x},\vec{\Ab},y)) \wedge \phi_\sort(z_\Aa,\vec{\Ab},\vec{x}))$ and we are done by part 1.

We have shown that both patterns have the same instances, which completes the proof.
\end{proof}

\begin{proof}[Proof of Prop.~\ref{prop:new-abs-conc}]
Using properties of $\new$:
\[\begin{array}{rcl}
(\new \Aa. \abs{\Aa}{\phi}) @ \Ab  
&=& \new \Aa. (\abs{\Aa}{\phi}) @ \Ab  ~~ \text{by Prop.~\ref{prop:new-wedge-not}(\ref{new-sigma})}\\
&=& \new \Aa. (\Ab \fresh \abs{\Aa}{\phi} \wedge \swap{\Aa}{\Ab}{\phi}) ~~\text{by Prop.~\ref{prop:abs-wedge-not}(\ref{abs-conc})}\\
&=& (\new \Aa. \Ab \fresh \abs{\Aa}{\phi}) \wedge (\new \Aa. \swap{\Aa}{\Ab}{\phi})~~\text{by Prop.~\ref{prop:new-wedge-not}(\ref{new-wedge})}\\
&=&  (\new \Aa. (\Ab = \Aa \vee \Ab \fresh \phi)) \wedge (\new \Aa. \swap{\Aa}{\Ab}{\phi})\\
&=& (\new \Aa. \Ab=\Aa \vee \new \Aa. \Ab \fresh \phi) \wedge (\new \Aa. \swap{\Aa}{\Ab}{\phi}) ~~\text{by Prop.~\ref{prop:new-wedge-not}(\ref{new-wedge})}\\
&=& (\new \Aa. \Ab \fresh \phi) \wedge (\new \Aa. \swap{\Aa}{\Ab}{\phi}) ~~\text{by Prop.~\ref{prop:new-reordered} (\ref{new-empty})}\\
&=&  \new \Aa. \swap{\Aa}{\Ab}{\phi}
\end{array}
\]
The last step follows from the semantics of freshness and $\new$: Given a valuation $\rho$, if $\semantics{\rho}{\Ab}$ is not fresh in an instance  $v \in \semantics{\rho[a/\Aa]}{[\Aa]\phi}$  for some new $a$ (i.e., the first conjunct is false), then $\semantics{\rho}{\Ab} \in \supp(v)$ for all new $a$. Therefore  $\semantics{\rho}{\Ab}$ is in the support of $\semantics{\rho[a/\Aa]}{\phi}$ for all new $a$, and hence $\new \Aa. \swap{\Aa}{\Ab}{\phi}$ has no instances. Otherwise, the first conjunct is true (interpreted as the whole carrier set) and the conjunction is therefore equivalent to the second conjunct.
\end{proof}

\end{document}